%% file: paper.tex
\documentclass[11pt]{article}

\usepackage{geometry}
\usepackage[T1]{fontenc}
\usepackage[utf8]{inputenc}

\usepackage{cleveref}
\usepackage{listings}
\usepackage{todonotes}
\usepackage{multicol}

\usepackage{titlesec}
\titlespacing*{\section}
{0pt}{0.75ex plus 1ex minus .2ex}{0.75ex plus .2ex}
\titlespacing*{\subsection}
{0pt}{0.5ex plus 1ex minus .2ex}{0.75ex plus .2ex}

\usepackage{algorithm}
\usepackage{algpseudocode}
\algrenewcommand\algorithmicindent{0.7em}

\usepackage{graphicx}
\usepackage{subcaption}
\usepackage{pgfplots}
\usepackage[inline]{enumitem}

\setlist{nosep}

\usepackage{amsmath}
\usepackage{amsthm}

\newtheorem{theorem}{Theorem}[section]
\newtheorem{lemma}[theorem]{Lemma}
\newtheorem{corollary}[theorem]{Corollary}

\makeatletter
\renewenvironment{proof}[1][\proofname]{\par
  \vspace{-\topsep}
  \pushQED{\qed}%
  \normalfont
  \topsep0pt \partopsep0pt 
  \trivlist
  \item[\hskip\labelsep
        \itshape
    #1\@addpunct{.}]\ignorespaces
}{%
  \popQED\endtrivlist\@endpefalse
  \addvspace{6pt plus 6pt} 
}
\makeatother

\newcommand{\euler}{e}

\newcommand{\cws}{Low Cost Work Stealing}
\newcommand{\ws}{Work Stealing}
\newcommand{\spdeque}{Split Deque}

\newcommand{\spd}{SpDeque}
\newcommand{\spds}{SpDeques}
\newcommand{\spdrs}{relaxed semantics}

\newcommand{\memfence}{Memory Fence}

\newcommand{\mfence}{MFence}

\newcommand{\compandswap}{Compare-And-Swap}
\newcommand{\cas}{CAS}

\title{Scheduling computations with provably low synchronization overheads}

\author{
  Guilherme Rito\\
  ETH-Zurich\\
  \texttt{guilherme.teixeira@inf.ethz.ch}
  \and
  Hervé Paulino\\
  FCT-UNL)\\
  \texttt{herve.paulino@fct.unl.pt}
}
\date{}

\begin{document}

\maketitle

\begin{abstract}
Work Stealing has been a very successful algorithm for scheduling parallel computations, and is known to achieve high performances even for computations exhibiting fine-grained parallelism.
%
%
We present a variant of \ws\ that provably avoids most synchronization overheads by keeping processors' deques entirely private by default, and only exposing work when requested by thieves.
This is the first paper that obtains bounds on the synchronization overheads that are (essentially) independent of the total amount of work, thus corresponding to a great improvement, in both algorithm design and theory, over state-of-the-art \ws\ algorithms.
Consider any computation with work $T_{1}$ and critical-path length $T_{\infty}$ executed by $P$ processors using our scheduler.
Our analysis shows that the expected execution time is $O\left(\frac{T_{1}}{P} + T_{\infty}\right)$, and the expected synchronization overheads incurred during the execution are at most $O\left(\left(C_{\cas} + C_{\mfence}\right)PT_{\infty}\right)$, where $C_{\cas}$ and $C_{\mfence}$ respectively denote the maximum cost of executing a Compare-And-Swap instruction and a Memory Fence instruction.
\end{abstract}


\clearpage

\input{sections/1-intro}

\input{sections/2-algorithm}

\input{sections/3-analysis}
\input{sections/4.1-comparison.tex}
\input{sections/5-conclusions-and-future-work}

\clearpage

\bibliographystyle{plain}
\bibliography{bibliography}

\clearpage

\input{sections/appendix-1}

\end{document}

%% file: sections/1-intro.tex
\section{Introduction}
\label{sec:intro}

\input{1-1-intro}


\input{1-2-related-work}

\input{1-3-contributions}

\input{sections/2-preliminaries}

%% file: 1-1-intro.tex

For some time now, the \ws\ algorithm is one of the most popular for scheduling multithreaded computations.
In \ws, each worker (usually referred to as \textit{processor}) owns a double-ended queue (deque) of threads ready to execute.
This deque is locally manipulated as a stack, similarly to a sequential execution:
processors push and pop threads from the bottom side of their deque when, respectively, a new thread is spawned and the execution of the current thread concludes.
Additionally, whenever a pop operation finds the local deque empty, the processor becomes a \emph{thief} and starts targeting other processors --- called its \emph{victims} --- uniformly at random, with the purpose of stealing a thread from the top of their deques.

As shown by Blumofe~\textit{et al.} in~\cite{DBLP:journals/jacm/BlumofeL99}, \ws\ is provably efficient for scheduling multithreaded computations.
However, due to the concurrent nature of processors' deques, the use of appropriate synchronization mechanisms is required to ensure correctness~\cite{DBLP:conf/popl/AttiyaGHKMV11}.
%
Consequently, even when processors are operating locally on their deques, they incur expensive synchronization overheads that, in most cases, are unnecessary.

The first provably efficient \ws\ algorithm, proposed by Blumofe \textit{et al.}~\cite{DBLP:journals/jacm/BlumofeL99}, assumed that all steal attempts targeting each deque were serialized, and only ensured the success of at most one such attempt per time step.
The idea was materialized in Cilk~\cite{DBLP:journals/jpdc/BlumofeJKLRZ96} via a blocking synchronization protocol named \emph{THE}.
Despite being extremely efficient, Frigo~\textit{et al.} found that the overheads introduced by the \textit{THE} protocol easily account for more than half of {Cilk}'s total execution time~\cite{DBLP:conf/pldi/FrigoLR98}.
Subsequent work mitigated part of these overheads by replacing the \textit{THE} protocol with a non-blocking one that resorts to \compandswap\ (\cas) and \memfence\ (\mfence) instructions~\cite{DBLP:journals/mst/AroraBP01,blumofe1998performance}.
Later, Morrison~\textit{et al.} tuned {Cilk} by removing a single \mfence\ instruction (one that was executed whenever a processor tried to take work from its deque) and found that this single \memfence\ could account for as much as 25\%~of the total execution time~\cite{DBLP:conf/asplos/MorrisonA14}.
Unfortunately, it has been proved, by Attiya~\textit{et al.} in~\cite{DBLP:conf/popl/AttiyaGHKMV11}, that it is impossible to eliminate all synchronization (e.g. the \mfence\ instruction mentioned above) from the implementation of any concurrent data structure that could possibly be used as a work-queue by a \ws\ algorithm, while maintaining correctness.
Indirectly, this result implies the impossibility of eliminating all synchronization from \ws\ algorithms that use any fully concurrent data structure as processor's work-queues.

Various proposals have been made with the goal of eliminating synchronization for local deque accesses, by making deques partly or even entirely private~\cite{DBLP:conf/ppopp/AcarCR13,DBLP:conf/sc/DinanLSKN09,DBLP:conf/ppopp/HiraishiYUY09,DBLP:conf/hpdc/LifflanderKK12,DBLP:conf/asplos/MorrisonA14,tzannes2012enhancing,DBLP:conf/europar/DijkP14}.
The elimination of synchronization for local deque accesses, however, raises a new problem.
Since synchronization is required to guarantee correctness, when a processor $p$ spawns a thread $\Gamma$ and pushes $\Gamma$ (locally) to its work-queue, $\Gamma$ cannot safely be stolen from $p$ by other processors, at least until $p$ issues some synchronization operation~\cite{DBLP:conf/popl/AttiyaGHKMV11}.
So, when should a busy processor use synchronization to permit load balancing?
The subtleness of this crucial question is evidenced by the current state-of-the-art: there is no algorithm that provably avoids most synchronization overheads while maintaining provably good performance.
%
%
On one hand, if a processor exposes work too eagerly, then it still incurs unnecessary synchronization overheads~ \cite{DBLP:conf/ppopp/AcarCR13,DBLP:conf/icpp/DinanKLNS08,DBLP:conf/sc/DinanLSKN09,DBLP:conf/hpdc/LifflanderKK12,tzannes2012enhancing,DBLP:conf/europar/DijkP14}.
On the other hand, if a processor barely exposes any work then load balancing opportunities become limited, thus potentially dropping the asymptotically optimal runtime guarantees of \ws~\cite{DBLP:conf/ppopp/HiraishiYUY09,DBLP:conf/asplos/MorrisonA14,DBLP:conf/europar/DijkP14}.
%
%
To address this problem optimally, our algorithm follows a lazy approach: (1) a processor $p$ only uses synchronization to expose work when a thief directly asks $p$ for work, and (2) $p$ only exposes a single unit of work (i.e. a single thread) for each time it is asked to expose work.
%

%% file: 1-2-related-work.tex
\subsection{Related work}
\label{sub:related-work}


Many efforts have been carried out towards reducing and even eliminating the expensive synchronization present in state-of-the-art \ws\ schedulers.

Michael~\textit{et al} proposes a variant of \ws\ for idempotent computations that reduces synchronization overheads by relaxing the semantics of work-queues from the conventional \textit{exactly-once} semantics to \textit{at-least-once} semantics~\cite{DBLP:conf/ppopp/MichaelVS09}.
By using work-queues that satisfy only the weaker semantics, processors no longer have to incur in expensive synchronization overheads when operating locally. 
Unfortunately, this approach (of relaxing the semantics of work-queues to \textit{at-least-once} semantics) not only is inherently limited, as it is only suitable for idempotent computations, but also drops the provably good performance guarantees of \ws.

%
%

Endo~\textit{et al.} was the first using split queues to avoid unnecessary synchronization~\cite{DBLP:conf/sc/EndoTY97}.
In this study, the authors present an implementation of a scalable garbage collector system that, by using clever load balancing techniques and split queues to avoid unnecessary synchronization overheads, achieves high performances even for large scale machines.
%
In~\cite{DBLP:conf/icpp/DinanKLNS08,DBLP:conf/sc/DinanLSKN09}, Dinan~\textit{et al.} studies \ws\ under a distributed environment and proposes the use of \spds\ to avoid synchronization for local deque accesses.
%
Lifflander~\textit{et al.} studies the execution of iterative over-decomposed applications~\cite{DBLP:conf/hpdc/LifflanderKK12} and proposes, among others, a message-based retentive \ws\ algorithm adapted for the execution of iterative workloads on large scale distributed settings.
%
Tzannes~\textit{et al.} proposes a scheduling algorithm where each processor keeps all of its work entirely private, except for the topmost node that is kept stored in a shared cell~\cite{tzannes2012enhancing}.
Since the algorithm always ensures that the topmost node is shared, it does not behave appropriately for computations in which processors frequently access the topmost nodes of their deques.
As mentioned in~\cite{DBLP:conf/ppopp/AcarCR13}, a similar limitation has been identified for the Chase-Lev Deque~\cite{DBLP:conf/spaa/ChaseL05}.
Unfortunately, in all these approaches (i.e. in~\cite{DBLP:conf/icpp/DinanKLNS08,DBLP:conf/sc/DinanLSKN09,DBLP:conf/hpdc/LifflanderKK12, tzannes2012enhancing}), processors expose work too eagerly, always leaving some work exposed for thieves to take.
A consequence of this design choice is that synchronization overheads still scale with the total amount of work.

Hiraishi~\textit{et al.} suggests that processors should behave as in a sequential execution~\cite{DBLP:conf/ppopp/HiraishiYUY09}.
Under their scheme, deques are kept entirely private and processors only permit parallelism when an idle processor requests work.
Upon such request, the busy processor backtracks to the last point where it could have spawned a task, spawns the task, offers it to the requesting processor, and then proceeds with the execution.
Since work requests are rare, the gains of eliminating synchronization for local operation can surpass the extra overheads arising from backtracking.
%

Morrison~\textit{et al.} studies alternative designs to the synchronization protocols used by \ws\ schedulers, considering the architectures of modern \emph{TSO} processors~\cite{DBLP:conf/asplos/MorrisonA14}.
In their algorithm, thieves can only steal work from a victim if such work is stored far enough from the bottom of the victim's deque to avoid any data race; this safe distance is computed \textit{a priori} by taking into account the size of the microprocessor's internal store buffer.
With this strategy, not only thieves can asynchronously take work from their victims, but processors can also access their deques locally without requiring any synchronization.
Unfortunately, such scheme suffers from a big limitation: the bottommost threads within a processor's deque cannot be stolen, and thus the scheduler is not appropriate for generic computations.


%
More recently, Dijk~\textit{et al.} studies the effectiveness of \spds\ on shared memory environments~\cite{DBLP:conf/europar/DijkP14,van2015certainty}.
In their approach, however, busy processors only check for work requests each time they access their deque, which precludes any performance guarantees for generic computations.
This is since the frequency at which busy processors may permit load balancing depends on the structure of the computation.
Moreover, whenever a busy processor realizes it was targeted by a steal attempt, it exposes at least half of its work.
This strategy increases the unnecessary synchronization costs incurred by the algorithm as busy processors now have to start accessing the shared part of their work-queue more often to fetch work. 
%

Acar~\textit{et al.} presents two \ws\ algorithms --- \textit{sender-} and \textit{receiver-initiated} --- that avoid synchronization by making deques entirely private to each processor~\cite{DBLP:conf/ppopp/AcarCR13}.
In addition to promising empirical results, the authors show that the expected execution time for both algorithms can be somewhat competitive with \ws\ algorithms that use concurrent deques.
Unfortunately, for the sender-initiated algorithm, busy processors now have to periodically search for idle ones, leading to unnecessary communication and synchronization overheads that still scale with the computation's execution time, and thus, indirectly, with the total amount of work.
The difference between the receiver-initiated algorithm and ours is more subtle, however.
In their receiver-initiated algorithm, busy processors now have to periodically check for incoming steal requests as well as to expose part of their current state by means of a flag that is periodically updated, thus requiring synchronization. 
This contrasts with our work, where there is no exposed state that processors periodically have to update.
This difference is reflected, for example, in a sequential execution: while our algorithm essentially does not use synchronization, the receiver-initiated algorithm does.

%% file: 1-3-contributions.tex
\subsection{Contributions}
\label{sub:contributions}


In this paper we present \cws, a variant of the \ws\ algorithm that uses \spds\ to provably avoid most synchronization overheads, while maintaining an asymptotically optimal expected runtime.
The theoretical significance of our contributions is highlighted, for instance, by the tight bounds we obtain on the synchronization overheads incurred by our algorithm.
Our bounds are essentially independent from the computation's total amount of work, thus contrasting with previous work.
From an algorithm design perspective, \cws\ greatly improves over prior \ws\ schedulers as it shows how to optimally use synchronization to permit provably efficient load balancing.
Four of the distinctive features of our algorithm are:
\begin{enumerate}
\item 
Busy processors only expose work to be stolen after being targeted by one or more steal attempts.
This allows processors to work locally on their work-queue without requiring any synchronization, imposing it only when load balancing may be required.

\item Work exposure requests are attended in constant time.
This is crucial to keep the algorithm's execution time bounds a constant factor away from optimal.
The requirement may be achieved by periodically checking for requests or by implementing an asynchronous notification mechanism.
For the sake of simplicity, we only focus on the former.

\item 
Processors only expose one thread of their local work at a time, contrasting with prior approaches. 
Only so, synchronization for local operation can be eliminated when load balancing is only sporadically required.

\item
All interactions between processors are completely asynchronous, making our algorithm viable for multiprogrammed environments.
  
\end{enumerate}

%
As we will see, our analysis shows that for a $P$-processor execution of a computation with total work $T_{1}$ and critical-path length (\textit{i.e.}~span) $T_{\infty}$, the expected runtime of \cws\ is at most $O\left(\frac{T_{1}}{P} + T_{\infty}\right)$, and the expected synchronization overheads incurred by the algorithm are at most $O\left(\left(C_{\cas} + C_{\mfence}\right)PT_{\infty}\right)$, where $C_{\cas}$ and $C_{\mfence}$ respectively denote the synchronization costs incurred by the execution of a \cas\ and \mfence\ instructions.
These bounds are tight and imply that for several classes of computations our algorithm reduces the use of synchronization by an almost exponential factor when compared with prior provably efficient \ws\ algorithms.

%
%
%

%

%
%
%
%

%% file: sections/2-preliminaries.tex
\subsection{Preliminaries}
\label{sub:preliminaries}


Like in much previous work~\cite{DBLP:journals/mst/AcarBB02,DBLP:conf/ppopp/AcarCR13,DBLP:conf/ipps/AgrawalHHL07,DBLP:journals/tocs/AgrawalLHH08, DBLP:conf/spaa/AroraBP98,DBLP:journals/mst/AroraBP01,DBLP:journals/jacm/BlumofeL99,DBLP:conf/spaa/MullerA16,DBLP:conf/isaac/TchiboukdjianGTRB10}, we model a computation as a \emph{dag} (\textit{i.e.} a direct acyclic graph) $G = \left (V, E \right)$,
where each node $v \in V$ corresponds to an instruction, and each edge $\left(\mu_{1},\mu_{2}\right) \in E$ denotes an ordering between two instructions (meaning $\mu_{2}$ can only be executed after $\mu_{1}$).
Nodes with in-degree of 0 are referred to as \emph{roots}, while nodes with out-degree of 0 are called \emph{sinks}.
Equivalently to Arora \textit{et al} in~\cite{DBLP:journals/mst/AroraBP01}, we make two assumptions related with the structure of computations.
Let $G$ denote a computation's dag:
\begin{enumerate*}
 \item there exists only one root and one sink in $G$;
 \item the out-degree of any node within $G$ is at most two (meaning that each instruction can spawn at most one thread).
\end{enumerate*}

\begin{figure*}
\centering
\begin{subfigure}{0.42\textwidth}
\centering
	\includegraphics[width=.6\linewidth]{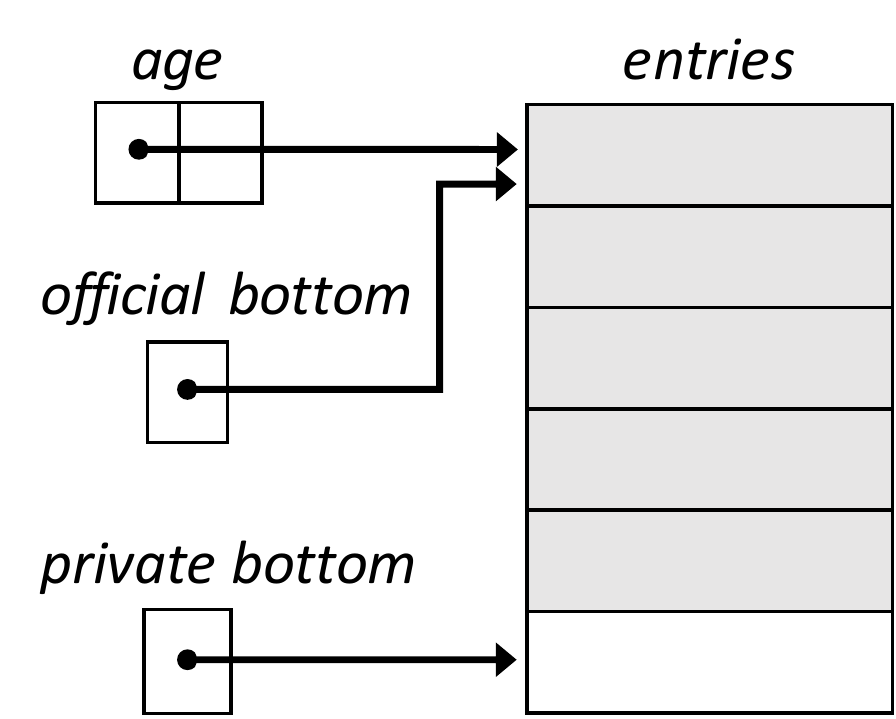}
\caption{An \spd\ with no stealable nodes}
\label{fig:spd-fpriv}
\end{subfigure}
\qquad \qquad \qquad
\begin{subfigure}{0.42\textwidth}
\centering
\includegraphics[width=.6\linewidth]{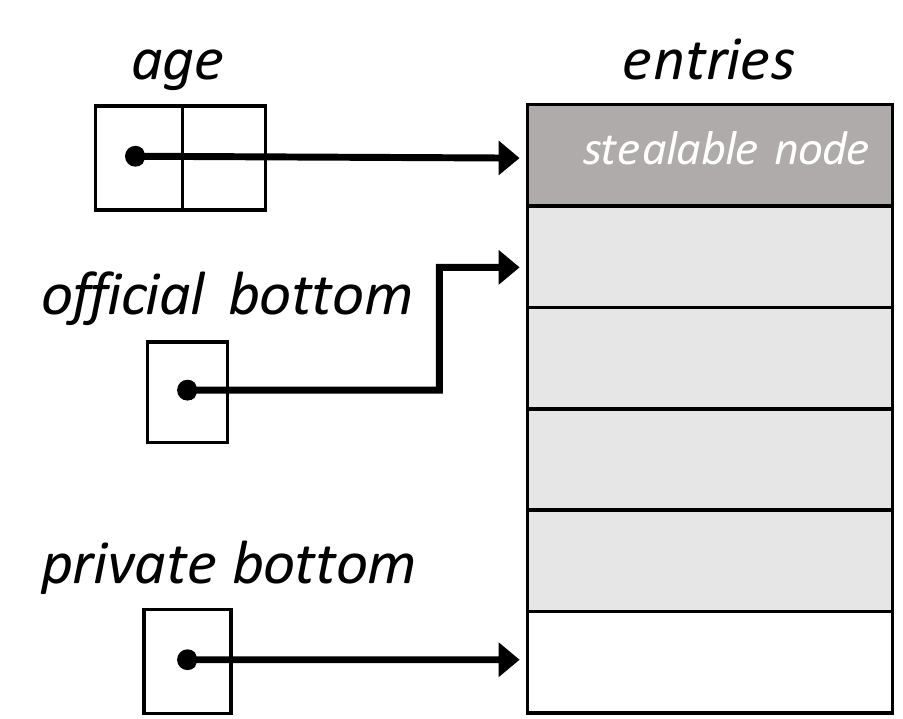}
\caption{An \spd\ with one stealable node}
\label{fig:spd-steal}
\end{subfigure}
\caption{The \spdeque}
\label{fig:spd}
\end{figure*}

The total number of nodes within a dag is expressed by $T_{1}$ and the length of a longest directed path (\textit{i.e.}~the critical-path length) by $T_{\infty}$.
A node is \emph{ready} if all its ancestors have been executed, implying that all the ordering constraints of $E$ are satisfied.
When a node becomes ready we say that it was \emph{enabled}; to ensure correction only ready nodes can be executed.
The assignment of a node $\mu$ to a processor $p$ means that $\mu$ will be the next node $p$ executes.
Finally, a computation's execution can be partitioned into discrete time steps, such that at each step, every processor executes an instruction.

%% file: sections/2-algorithm.tex
\section{\cws}
\label{sec:algorithm}

In \cws, each processor owns a Lock Free \spd, instead of a typical concurrent deque.
An \spd\ (illustrated in Figure~\ref{fig:spd}) is simply a deque that is split into two parts: a private part and a public part.
The public part lies in the top of the \spd\ whereas the private corresponds to the rest of the \spd.
To avoid synchronization for local operations, only the owner of an \spd\ is allowed to access its private part.
Furthermore, by default busy processors operate on the private part of their \spd, pushing and popping ready nodes as necessary.
In fact, a busy processor only attempts to fetch work from the public part of its \spd\ if the private part is empty.
In such situation, if the processor's attempt succeeds (\textit{i.e.}~if the public part of the processor's \spd\ is not empty), the obtained node becomes the processor's new assigned node.
However, if the public part of the processor's \spd\ is also empty, the processor becomes a thief and begins a stealing phase.
During stealing phases, thieves target victims uniformly at random and attempt to steal work from the top of their \spds.
To keep the private part of \spds\ entirely private, steal attempts are only allowed to access the public part.
Thus, when a thief attempts to steal work from a victim's \spd\ whose public part is empty (illustrated in Figure~\ref{fig:spd-fpriv}), the steal attempt simply fails and the thief does not obtain work.
In such case, the thief then updates a victim's flag (referred to as the $targeted$ flag) to (asynchronously) notify the victim that the public part of its \spd\ is empty (more on this ahead).
When the owner of the \spd\ realizes it was notified (by checking the value of its $targeted$ flag), it tries to transfer a node from the private part of its \spd\ to the public part.
If the private part is not empty then a node is transferred, in which case we say that the transferred node became \emph{stealable} (illustrated in Figure~\ref{fig:spd-steal}).


\subsection{The Lock-Free \spdeque}
\label{sub:spdeque}

We now present the specification of an \spd\ object, along with its associated \spdrs.
Being the behavior of \spds\ similar to the behavior of concurrent deques, the \spd's relaxed semantics are comparable to the relaxed deque semantics presented in~\cite{DBLP:journals/mst/AroraBP01}.
An \spd\ object meeting the \spdrs\ supports five methods:
\begin{description}
  \item \textsl{push} --- Pushes a node into the bottom of the \spd's private part.

  \item \textsl{pop} --- Removes and returns a node from the bottom of the \spd's private part, if that part is not empty.
  Otherwise, returns the special value \textsc{race}.

  \item \textsl{updateBottom} --- Transfers the topmost node from the private part of the \spd\ into the bottom of the public part, and does not return a value.
  The invocation has no effect if the private part of the \spd\ is empty.

  \item \textsl{popBottom} --- Removes and returns the bottom-most node from the public part of the \spd.
  If the \spd\ is empty, the invocation has no effect and \textsc{empty} is returned.

  \item \textsl{popTop} --- Attempts to remove and return the topmost node from the public part of the \spd.
  If the public part is empty, the invocation has no effect and the value \textsc{empty} is returned.
  If the invocation aborts, it has no effect and the value \textsc{abort} is returned.
\end{description}

An \spd\ implementation is constant-time \textit{iff} any invocation to each of these methods takes at most a constant number of steps to return.
Say that a set of invocations to an \spd's methods meets the \spdrs\ \textit{iff} there is a set of \emph{linearization times} for the corresponding non-aborting invocations such that:
\begin{enumerate}
  \item Every non-aborting invocation's linearization time lies within the beginning and completion times of the respective invocation;
  \item No linearization times associated with distinct non-aborting invocations coincide;
  \item The return values for the non-aborting invocations are consistent with a serial execution of the methods in the order given by the linearization times of the corresponding non-aborting invocations; and
  \item For each aborted \textsl{popTop} invocation $x$ to an \spd\ $d$, there exists another invocation removing the topmost item from $d$ whose linearization time falls between the beginning and completion times of invocation $x$.
\end{enumerate}


\subsection{The \cws\ Algorithm}
\label{sub:cws}

Algorithm~\ref{algo:cws} depicts the specification of the \cws\ algorithm.
Each processor owns an \spd\ that uses to store its attached nodes and, additionally, owns a $targeted$ flag that stores a Boolean value.
This flag is used to implement an asynchronous notification mechanism that allows thieves to request their victims to expose work, allowing it to be stolen.
Even though, in practice, the notification mechanism our algorithm can be implemented using signals, to perform a correct analysis of the algorithm's synchronization overheads all the possible sources of such overheads must be explicit, for which reason we chose to embed a simple notification mechanism into the algorithm's specification.
%
%
Although the $targeted$ flag of each processor can be simultaneously accessed by multiple processors, to ensure the algorithm's correctness it suffices to guarantee that no read nor write operation to a processor's $targeted$ flag is cached.

\begin{algorithm}
\caption{The \cws\ algorithm.}
\label{algo:cws}
\vspace{-15pt}
\begin{multicols}{2}
\begin{scriptsize}
   \begin{algorithmic}[1]
    \Procedure{Scheduler}{}
      \While{$computation \ not \ terminated$}
        \If{$self.targeted$}
          \State $	self.spdeque$.updateBottom() 
          \State $self.targeted \gets $ \textsc{false}
        \EndIf
        \If{ValidNode($assigned$)}
          \State $enabled \gets $execute($assigned$)
          \If{length($enabled$) $> 0$}
            \State $assigned \gets enabled\left[0\right]$
            \If{length($enabled$) $= 2$}
              \State $self.spdeque$.push($enabled\left[1\right]$)
            \EndIf
          \Else
            \State $assigned \gets self.spdeque$.pop()
            \If{$assigned = $ \textsc{race}}
              \State $assigned \gets self.spdeque$.popBottom()
            \EndIf
          \EndIf
        \Else
          \State $self$.WorkMigration()
        \EndIf
      \EndWhile
    \EndProcedure
\newline
    \Procedure{WorkMigration}{}
      \State $victim \gets $ UniformlyRandomProcessor()
      \State $assigned \gets victim.spdeque$.popTop()
      \If{$assigned = $ \textsc{empty}}
        \State $victim.targeted \gets $ \textsc{true}
      \EndIf
    \EndProcedure
\newline
    \Function{ValidNode}{$node$}
      \State \Return~$node\,\neq\,$\textsc{empty}
      {$\quad and \quad node\,\neq\,$ \textsc{abort}}
      {$\quad and \quad node\,\neq\,$ \textsc{none}}
    \EndFunction
  \end{algorithmic}
\end{scriptsize}
\end{multicols}
\vspace{-10pt}
\end{algorithm}

Before a computation's execution begins, every processor sets its $assigned$ node to \textsc{none} and its $targeted$ flag to \textsc{false}.
To start the execution, one of the processors gets the root node assigned.

As we will see, the behavior of \cws\ is similar to the original \ws\ algorithm.
Consider some processor $p$ working on a computation scheduled by \cws, and some iteration of the scheduling loop that $p$ executes (corresponding to lines 2 to 23 of Algorithm~\ref{algo:cws}).
First, $p$ reads the value of its $targeted$ flag to check if it has been notified by some thief.
If $p$'s $targeted$ flag is set to \textsc{true} (\textit{i.e.}~if $p$ was notified), the processor tries to make a node stealable, by invoking \textsl{updateBottom} to its \spd.
After that, and regardless of that invocation's outcome, $p$ resets its $targeted$ flag back to \textsc{false}.
The subsequent behavior of $p$ depends on whether it has an assigned node.

\begin{itemize}
  \item If $p$ has an assigned node, $p$ executes the node.
    From this execution, either zero, one or two nodes can be enabled.
      \begin{description}

        \item[Zero nodes enabled] The processor tries to fetch the bottommost node stored in its \spd.
          To that end, $p$ first tries to fetch a node from the bottom of its \spd's private part (line 15).
          If $p$ finds that part empty, it then tries to fetch a node from the public part (line 17).
          If this part is also empty, $p$ becomes a thief and starts a work stealing phase.
          On the other hand, if $p$ successfully fetched a node from any of the parts of its \spd, then the returned node becomes $p$'s new assigned node.

        \item[One node enabled] The enabled node becomes $p$’s new assigned node (line 10).

        \item[Two nodes enabled] One of the enabled nodes becomes $p$'s new assigned node, whilst the other is pushed into the bottom of the private part of $p$'s \spd\ (line 12).
      \end{description}

  \item If $p$ does not have an assigned node, it is searching for work. 
    In this situation, the processor first targets, uniformly at random, a victim processor and then attempts to steal work from the public part of the victim's \spd\ (lines 26 and 27).
    If the attempt is successful, the stolen node becomes $p$'s new assigned node.
    If the attempt aborts, $p$ simply gives up on the steal attempt.
    For last, if $p$ finds the public part of the victim's \spd\ empty it sets the victim's $targeted$ flag to \textsc{true} (line 29), notifying the victim that it found the public part of the victim's \spd\ empty.
\end{itemize}


\subsection{A \spdeque\ Implementation}
\label{sub:spdeque-imp}
Algorithm~\ref{algo:spdeque} depicts a possible implementation of the lock-free \spd, based on the deque's implementation given in~\cite{DBLP:journals/mst/AroraBP01}.
As illustrated in Figure~\ref{fig:spd}, each \spd\ object has four instance variables:
\begin{description}
  \item $entries$ ---  an array of ready nodes.
  \item $privateBottom$ --- the index below the bottommost node of the \spd.
  \item $\textit{officialBottom}$ --- the index below the bottommost node of the \spd's public part.
  \item $age$ ---  composed of two fields:  $top$, which corresponds to the index right below the topmost node of the \spd's public part; and $tag$, which is only used to ensure correction (avoiding the famous \emph{ABA} problem).
\end{description}


\begin{algorithm}[!h]
\caption{The \spd\ implementation}
\label{algo:spdeque}
\begin{scriptsize}
\vspace{-15pt}
\begin{multicols}{2}
   \begin{algorithmic}[1]
\Statex    $privateBottom \gets 0$ // {private field}
\Statex    $entries \gets \{\}$ // private read-write, public read-only
\Statex    ${officialBottom} \gets 0$  // private read-write, public read-only
\Statex   $age \gets \{0, 0\}$  // {public field}
  \newline
   \setcounter{ALG@line}{0}
    \Procedure{push}{$node$}
      \State $pBot \gets self.privateBottom$
      \State $self.entries[pBot] \gets node$
      \State $self.privateBottom \gets pBot + 1$
    \EndProcedure
  \newline
    \Procedure{pop}{}
      \State $pBot \gets self.privateBottom$
      \If{$pBot = self.{officialBottom}$}
         \Return \textsc{race}
      \EndIf
      \State $pBot \gets pBot - 1$
      \State $node \gets self.entries[pBot]$
      \State $self.privateBottom \gets pBot$
      \State \Return $node$
    \EndProcedure
  \newline
    \Procedure{popTop}{}
      \State $oldAge \gets self.age$
      \State $oldBottom \gets self.{officialBottom}$
      \If{$oldBottom \leq oldAge.top$}
         \Return \textsc{empty}
      \EndIf
      \State $node \gets self.entries[oldAge.top]$
      \State $newAge \gets oldAge$
      \State $newAge.top \gets newAge.top + 1$
      \If{CAS($age$, $oldAge$, $newAge$) = \textsc{success}}
        \State \Return $node$
      \EndIf
      \State \Return \textsc{abort}
    \EndProcedure
    \newline
    \Procedure{updateBottom}{}
    \State $pBot \gets self.privateBottom$
    \State $oBot \gets self.{officialBottom}$
    \If{$pBot > oBot$}
       $oBot \gets oBot + 1$
    \EndIf
    \State $self.{officialBottom} \gets oBot$
  \EndProcedure
\newline
      \Procedure{popBottom}{}
        \State $oBot \gets self.{officialBottom}$
        \If{$oBot = 0$}
           \Return \textsc{empty}
        \EndIf
        \State $oBot \gets oBot - 1$
        \State $self.{officialBottom} \gets oBot$
        \State $node \gets self.entries[oBot]$
        \State $oldAge \gets age$
        \If{$oBot > oldAge.top$}
           \Return $node$
        \EndIf
        \State $self.{officialBottom} \gets 0$
        \State $self.privateBottom \gets 0$
        \State $newAge.top \gets 0$
        \State $newAge.tag \gets oldAge.tag + 1$
        \If{$oBot = oldAge.top$}
          \If{CAS($age$, $oldAge$, $newAge$) = \textsc{success}}
            \State \Return $node$
          \EndIf
        \EndIf
        \State $self.age \gets newAge$
        \State \Return \textsc{empty}
      \EndProcedure
  \end{algorithmic}
  \end{multicols}
\end{scriptsize}
\vspace{-10pt}
\end{algorithm}

We say that a set of invocations is \emph{good} if and only if the methods \textsl{push}, \textsl{pop}, \textsl{updateBottom} and \textsl{popBottom} are never invoked concurrently.
For \cws, as only the owner of each \spd\ can invoke these methods, it is easy to deduce that all sets of invocations issued by the algorithm are good.
Furthermore, we claim that the implementation depicted in Algorithm~\ref{algo:spdeque} is constant-time and meets the \spdrs\ (defined in \Cref{sub:spdeque}) on any good set of invocations.
However, even though all methods are composed by a small number of instructions and none includes a loop, proving this claim is not a straightforward task because all possible execution interleaves have to be considered.
Moreover, as the main focus of this study is not related with programs' verification, the proof of this claim falls out of the scope of this paper.
Yet, we remark that the proposed implementation is a simple extension of the deque implementation presented in~\cite{DBLP:journals/mst/AroraBP01}, which has been proven in~\cite{blumofe1999verification} to be a correct implementation, meeting the relaxed deque semantics on any set of invocations made by the \ws\ algorithm.
For this reason, throughout this paper we assume that for any set of invocations issued by the \cws\ algorithm, the \spdrs\ is always satisfied.

\begin{lemma}
  \label{lemma:syncfree-push}
  No invocation to \textsl{push} requires a \mfence\ instruction.
\end{lemma}
\begin{proof} 
  Since the \textsl{push} method operates only once over a single publicly accessible field ($entries$) of the \spd's state, no \mfence\ instructions are required.
\end{proof}

\begin{lemma}
  \label{lemma:syncfree-pop}
  No invocation to \textsl{pop} requires a \mfence\ instruction.
\end{lemma}
\begin{proof} 
  Any invocation to the \textsl{pop} method only reads from two publicly accessible fields of the \spd's state, namely $\textit{officialBottom}$\ (line 8) and $entries$\ (line 12).
  However, due to a data dependency, no re-ordering between these read operations may occur, and so, no \mfence\ instructions are required.
\end{proof}

The dag of a computation is dynamically unfolded during its execution.
If the execution of a node $u$ \emph{enables} another node ${u}'$, then $(u,{u}')$ is an \emph{enabling edge} and refer to node $u$ as the \emph{designated parent} of ${u}'$.
Refer to the tree formed by the enabling edges of a particular execution of a dag by \emph{enabling tree}, and denote the depth of a node $u$ within this tree by $d\left(u\right)$.
Define the weight of $u$ as $w\left(u\right) = T_{\infty} - d\left(u\right)$.
Similarly to~\cite{DBLP:journals/mst/AroraBP01}, our analysis is made in an \textit{a posteriori} fashion, allowing us to refer to the enabling tree generated by a computation's execution.


The following corollary is a direct consequence of the standard properties of deques (a full proof can be found in the appendix (Lemma~\ref{lemma:app-structural lemma})).

\begin{corollary}
\label{corollary:structural corollary}
Let $v_{1},\ldots,v_{k}$ denote the nodes stored in some processor $p$'s \spd, ordered from the bottom of the \spd\ to the top, at some moment during the execution of \cws.
Moreover, let $v_{0}$ denote $p$'s assigned node (if any).
Then, we have $w\left(v_{0}\right) \leq w\left(v_{1}\right) < \ldots < w\left(v_{k-1}\right) < w\left(v_{k}\right)$.
\end{corollary}



%% file: sections/3-analysis.tex
\section{Analysis}
\label{sec:analysis}
In this section we obtain bounds on the expected execution time of computations using \cws, and on the expected synchronization overheads incurred by the scheduler.
The analysis we make follows the same overall idea as the one given in~\cite{DBLP:journals/mst/AroraBP01}.
Due to space restrictions, it is not possible to include all the proofs in the paper (which are thus presented in the appendix).
Before advancing any further, we introduce a few more essential definitions. 

Define a \emph{scheduling iteration} as a sequence of instructions executed by a processor corresponding to a particular iteration of the scheduling loop (lines 2 to 23 of Algorithm~\ref{algo:cws}).
Thus, the full sequence of instructions executed by each processor during a computation's execution can be partitioned into scheduling iterations.
As in~\cite{DBLP:journals/mst/AroraBP01}, we introduce the concept of \emph{milestone}:
an instruction within the sequence executed by a processor is a milestone \textit{iff} it corresponds to a node's execution (line 8) or to the return of a call to \textsl{WorkMigration} (line 31).
Taking into account the definition of a scheduling iteration it is clear that any scheduling iteration of the algorithm includes a milestone.
Refer to iterations whose milestone corresponds to a node's execution as \emph{busy iterations}, and refer to the remainder as \emph{idle iterations}.
As one might note, if a processor has an assigned node at the beginning of an iteration's execution, the iteration is a busy one, and, otherwise, the iteration is an idle one.
%
%
By observing the scheduling loop (lines 2 to 23 of Algorithm~\ref{algo:cws}), and taking into account that the \spd's implementation is constant time, it is clear that any scheduling iteration is composed of a constant number of instructions.
It then follows that any processor executes at most a constant number of instructions between two consecutive milestones.
Throughout the analysis, let $C$ denote a constant that is large enough to guarantee that any sequence of instructions executed by a processor with length at least $C$ includes a milestone.


We can now bound the execution time of a computation depending on the number of idle iterations that take place during that computation's execution.
The proof of the following result can be found in the appendix (Lemma~\ref{lemma:app-runtime-by-iterations}), and is a trivial variant of~\cite[Lemma 5]{DBLP:journals/mst/AroraBP01}, but considering the \cws\ algorithm.

\begin{lemma}
\label{lemma:runtime-by-iterations}
Consider any computation with work $T_{1}$ being executed by $P$ processors, under \cws.
The execution time is $O\left(\frac{T_{1}}{P} + \frac{I}{P}\right)$, where $I$ denotes the number of idle iterations executed by processors.
\end{lemma}

As we will see, the following two results are key, as they show that the synchronization overheads incurred by \cws\ (essentially) only depend on the number of idle iterations that take place during a computation's execution (proofs in appendix (Lemmas~\ref{lemma:app-syncfree-iteration} and~\ref{lemma:app-cas-and-mem-barriers-by-throws}, respectively)).

\begin{lemma}
\label{lemma:syncfree-iteration}
Consider a processor $p$ executing a busy iteration such that $p$'s $targeted$ flag is set to \textsc{false} when $p$ checks it at the beginning of the iteration.
If the execution of $p$'s assigned node enables one or more nodes, or, if the private part of $p$'s \spd\ is not empty, then, no \mfence\ instruction is required during the execution of the iteration.
\end{lemma}


\begin{lemma}
\label{lemma:cas-and-mem-barriers-by-throws}
Consider any computation being executed by the \cws\ algorithm, using $P$ processors.
The number of \cas\ and \mfence\ instructions executed by processors during the computation's execution is at most $O\left(I + P\right)$, where $I$ denotes the total number of idle iterations executed by processors.
\end{lemma}

\subsection{Bounds on the expected number of idle iterations}
\label{sub:bounds-expected-idle-iterations}

The rest of the analysis focuses on bounding the number of idle iterations that take place during a computation's execution, and follows the same general arguments as the analysis presented in~\cite{DBLP:journals/mst/AroraBP01}.


We say that a node $u$ is \emph{stealable} if $u$ is stored in the public part of some processor's \spd.
Furthermore, we denote the set of ready nodes at some step $i$ by $R_{i}$.
Consider any node $u \in R_{i}$.
The potential associated with $u$ at step $i$ is denoted by $\phi_{i}\left(u\right)$ and is defined as
\[\phi_{i}\left(u\right) =
\begin{cases}
4^{3w\left(u\right) - 2} & \text{if\,}u\text{\,is\,assigned}\\
4^{3w\left(u\right) - 1} & \text{if\,}u\text{\,is\,stealable}\\
4^{3w\left(u\right)} & \text{otherwise}
\end{cases}\]

The total potential at step $i$, denoted by $\Phi_{i}$, corresponds to the sum of potentials of all the nodes that are ready at that step: $\Phi_{i} = \sum_{u \in R_{i}}\phi_{i}\left(u\right)$.

The following lemma is a formalization of the arguments already given in~\cite{DBLP:journals/mst/AroraBP01}, but considering the potential function we present (proof in appendix (Lemma~\ref{lemma:app-potential properties})).

\begin{lemma}
\label{lemma:potential properties}
Consider some node $u$, ready at step $i$ during the execution of a computation.
\begin{enumerate}
  \item If $u$ gets assigned to a processor at that step, the potential drops by at least $\frac{3}{4}\phi_{i}\left(u\right)$.
  \item If $u$ becomes stealable at that step, the potential drops by at least $\frac{3}{4}\phi_{i}\left(u\right)$.
  \item If $u$ was already assigned to a processor and gets executed at that step $i$, the potential drops by at least $\frac{47}{64}\phi_{i}\left(u\right)$.
\end{enumerate}
\end{lemma}


For the remainder of the analysis, we make use of a few more definitions, first introduced in~\cite{DBLP:journals/mst/AroraBP01}.
We denote the set of ready nodes attached to some processor $p$\ (\textit{i.e.} the ready nodes in $p$'s \spd\ together with the node it has assigned, if any) at the beginning of some step $i$ by $R_{i}\left(p\right)$.
Furthermore, we define the total potential associated with $p$ at step $i$ as the sum of the potentials of each of the nodes that is attached to $p$ at the beginning of that step
$\Phi_{i}\left(p\right) = \sum_{u \in R_{i}\left(p\right)}\phi_{i}\left(u\right)$.

For each step $i$, we partition the processors into two sets, $D_{i}$ and $A_{i}$, where the first is the set of all processors whose \spd\ is not empty at the beginning of step $i$ while the second is the set of all other processors (\textit{i.e.} the set of all processors whose \spd\ is empty at the beginning of that step).
Thus, the potential of any step $i$, $\Phi_{i}$, is composed by the potential associated with each of these two partitions
$\Phi_{i}= \Phi_{i}\left(D_{i}\right) + \Phi_{i}\left(A_{i}\right)$,
where
$\Phi_{i}\left(D_{i}\right) = \sum_{p \in D_{i}}\Phi_{i}\left(p\right)$ and $\Phi_{i}\left(A_{i}\right) = \sum_{p \in A_{i}}\Phi_{i}\left(p\right)$.


The following lemma is a direct consequence of Corollary~\ref{corollary:structural corollary} and of the potential function's properties (proof in appendix (Lemma~\ref{lemma:app-top-heavy spdeques})).

\begin{lemma}
\label{lemma:top-heavy spdeques}
Consider any step $i$ and any processor $p \in D_{i}$.
The top-most node $u$ in $p$'s \spd\ contributes at least $\frac{4}{5}$ of the potential associated with $p$.
That is, we have $\phi_{i}\left(u\right) \geq \frac{4}{5}\Phi_{i}\left(p\right)$.
\end{lemma}


With this, we now show that if a processor $p$ is targeted by a steal attempt, then $p$'s potential decreases by a constant factor (proof in appendix (Lemma~\ref{lemma:app-potential-decrease})).

\begin{lemma}
\label{lemma:potential-decrease}
Suppose a thief processor $p$ chooses a processor $q \in D_{i}$ as its victim at some step $j$, such that $j \geq i$\ (\textit{i.e.}~a steal attempt of $p$ targeting $q$ occurs at step $j$).
Then, at step $j + 2C$, the potential decreased by at least $\frac{3}{5}\Phi_{i}\left(q\right)$ due to either the assignment of the topmost node in $q$'s \spd, or for making the topmost node of $q$'s \spd\ become stealable.
\end{lemma}


The next lemma is a trivial generalization of the original result presented in~\cite[Balls and Weighted Bins]{DBLP:journals/mst/AroraBP01}
(proof in appendix (Lemma~\ref{lemma:app-balls and weighted bins})).

\begin{lemma}[Balls and Weighted Bins]
\label{lemma:balls and weighted bins}
Suppose we are given at least $B$ balls, and exactly $B$ bins.
Each of the balls is tossed independently and uniformly at random into one of the $B$ bins, where for $i = 1,\ldots,\,B$, bin $i$ has a weight $W_{i}$.
The total weight is $W = \sum_{i = 1}^{B} W_{i}$.
For each bin $i$, we define the random variable $X_{i}$ as \[X_{i} = \left\{\begin{matrix}
W_{i} & \text{if some ball lands in bin } i\\
0 & \text{otherwise}
\end{matrix}\right.\]
and define the random variable $X$ as $X = \sum_{i = 1}^{B} X_{i}$.
Then, for any $\beta$ in the range $0 < \beta < 1$, we have $P\left\{X \geq \beta W\right\} \geq 1 - \frac{1}{\left(1 - \beta\right)\euler}$.
\end{lemma}


The following result states that for each $P$ idle iterations that take place, with constant probability, the total potential drops by a constant factor.
The result is a consequence of Lemmas~\ref{lemma:potential-decrease} and~\ref{lemma:balls and weighted bins}
(proof in appendix (Lemma~\ref{lemma:app-phase potential decrease})).

\begin{lemma}
\label{lemma:phase potential decrease}
Consider any step $i$ and any later step $j$ such that at least $P$ idle iterations occur from $i$\ (inclusive) to $j$\ (exclusive).
Then, we have $P\left\{\Phi_{i} - \Phi_{j + 2C} \geq \frac{3}{10}\Phi_{i}\left(D_{i}\right)\right\} > \frac{1}{4}$.
\end{lemma}


Following  Lemma~\ref{lemma:phase potential decrease}, we are  able to bound the expected number of idle iterations that take place during a computation's execution using the \cws\ algorithm
(proof in appendix (Lemma~\ref{lemma:app-bounded-idle-iterations})).

\begin{lemma}
\label{lemma:bounded-idle-iterations}
Consider any computation with work $T_{1}$ and critical-path length $T_{\infty}$ being executed by \cws\ using $P$ processors.
The expected number of idle iterations is at most $O\left(P T_{\infty}\right)$, and with probability at least $1 - \varepsilon$, the number of idle iterations is at most $O\left(\left(T_{\infty} + \ln\left(\frac{1}{\varepsilon}\right)\right)P\right)$.
\end{lemma}

Finally, using Lemma~\ref{lemma:bounded-idle-iterations}, we can obtain bounds on both expected runtime of computations executed by the \cws\ algorithm, and the associated synchronization overheads.
\begin{theorem}
\label{theorem:bounded-execution time-memory-fences issued}
Consider any computation with work $T_{1}$ and critical-path length $T_{\infty}$ being executed by the \cws\ algorithm with $P$ processors.
The expected execution time is at most $O\left(\frac{T_{1}}{P} + T_{\infty}\right)$, and with probability at least $1 - \varepsilon$, the execution time is at most $O\left(\frac{T_{1}}{P} + T_{\infty} + \ln\left(\frac{1}{\varepsilon}\right)\right)$.
Moreover, the expected number of \cas\ and \mfence\ instructions executed during the computation's execution caused by \cws\ is at most $O\left(P T_{\infty}\right)$, and with probability at least $1 - \varepsilon$ the number of \cas\ and \mfence\ instructions executed is at most $O\left(P\left(T_{\infty} + \ln\left(\frac{1}{\varepsilon}\right)\right)\right)$.
\end{theorem}
\begin{proof}
Both results follow directly from Lemmas~\ref{lemma:runtime-by-iterations}, \ref{lemma:cas-and-mem-barriers-by-throws} and~\ref{lemma:bounded-idle-iterations}.
\end{proof}


\begin{corollary}
\label{corollary:bounded-sync-overheads}
Consider the statement of Theorem~\ref{theorem:bounded-execution time-memory-fences issued}.
Furthermore, let $C_{\cas}$ and $C_{\mfence}$ denote, respectively, the maximum synchronization overheads incurred by the execution of a \cas\ and \mfence\ instructions.
The expected synchronization overheads incurred by \cws\ are at most \[O\left(\left(C_{\cas} + C_{\mfence}\right)PT_{\infty}\right),\] and with probability at least $1 - \varepsilon$ the synchronization overheads incurred by \cws\ are at most $O\left(\left(C_{\cas} + C_{\mfence}\right)P\left(T_{\infty} + \ln\left(\frac{1}{\varepsilon}\right)\right)\right)$.
\end{corollary}
\begin{proof}
As already mentioned, and by considering the definition of \cws, depicted in Algorithm~\ref{algo:cws}, the only synchronization mechanisms the scheduler uses are \cas\ and \mfence\ instructions.
This corollary is then a direct consequence of Theorem~\ref{theorem:bounded-execution time-memory-fences issued} that takes into account the maximum possible overhead incurred by the execution of a single \cas\ and \mfence\ instructions.
\end{proof}

%% file: sections/4.1-comparison.tex
\section{Comparison with Work Stealing}
\label{sub:comparison}
To get a better understanding of the importance of avoiding synchronization for local deque accesses, we now compare the synchronization costs of our algorithm against conventional \ws\ algorithms that use concurrent deques.
To that end, we developed a simulator that, given a computation's dag, executes it, monitoring not only the number of \cas\ and \mfence\ instructions executed but also the number times that thieves requested other processors to expose work.
In this section we use the term \textit{notification} to refer to when a thief sets another processor's $targeted$ flag to \textsc{true}, requesting it to expose work.

\input{sections/plots}

For this comparison, we consider two distinct classes of dags: \emph{regular} and \emph{irregular}.
Regular dags  essentially correspond to trees of instructions where every non-leaf instruction forks two other instructions, and whose depth is given by an argument that is passed to the simulator.
Irregular Dags are intended to simulate unbalanced computations.
To that end, we use the argument passed to the simulator as the total depth of the dag, and make the depth between each two consecutive fork instructions follow an exponential distribution with parameter $\lambda = 0.05$.
The first class of dags corresponds to computations with balanced parallelism (\textit{e.g.}~Fibonacci, Parallel-For, \textit{etc}) whist the second corresponds to the ones with unbalanced parallelism (\textit{e.g.}~Graph Searches).

From Figure~\ref{fig:rec-nr-ops-span}, it is clear that while for \ws\ with concurrent deques the number of synchronization operations grows linearly with the total amount of work (and thus exponentially increases with the span of the dag), for \cws\ the number of synchronization operations and notifications scales linearly with the span of the computation.
Thus, even if the costs of handling notifications (\textit{i.e.}~of exposing work) were a thousand times greater than the cost of executing \cas\ or \mfence\ instructions, for dags with fork-span of at least $\approx 20$, our algorithm would incur in less synchronization overheads.
In practice, computations with a fork-span $\geq 20$ are very common, especially among fine-grained parallelism.
Unfortunately, due to the limitations that come with using a simulator, we have not been able to benchmark dags with a fork-span greater than $25$.
Yet, we remark that the trend is obvious and confirms that the use of \spds\ allows to avoid most of the synchronization that is present in conventional \ws\ algorithms.
Figure~\ref{fig:unb-nr-ops-span} reinforces our insight, showing that even for computations exhibiting irregular parallelism, \cws\ is able to avoid most of the synchronization present in \ws\ algorithms that use concurrent deques.
Finally, Figure~\ref{fig:rec-nr-ops-procs} shows that, while the synchronization costs for \ws\ are always extremely high, even for single processor executions, for \cws\ these costs only grow linearly with the number of processors used and, for a single processor execution, synchronization is negligible.

From a more theoretical perspective, note that, by taking into account our assumptions (which are standard~\cite{DBLP:journals/mst/AcarBB02,DBLP:conf/ppopp/AcarCR13,DBLP:conf/ipps/AgrawalHHL07,DBLP:journals/tocs/AgrawalLHH08, DBLP:conf/spaa/AroraBP98,DBLP:journals/mst/AroraBP01,DBLP:journals/jacm/BlumofeL99,DBLP:conf/spaa/MullerA16,DBLP:conf/isaac/TchiboukdjianGTRB10}) regarding computations' structure, we can create computations for which $T_{1} = O\left(2^{T_{\infty}}\right)$ (which correspond to dags of the first class).
Since for such computations the number of deque accesses is directly proportional to the total amount of work ($T_{1}$), our result shows that the use of \spds\ allows to reduce by almost an exponential factor the synchronization present in conventional \ws\ algorithms.

%% file: sections/plots.tex
\begin{figure*}
\centering
\begin{subfigure}{0.32\textwidth}
\centering
\resizebox{\linewidth}{!}{%
\begin{tikzpicture}
\begin{semilogyaxis}[
    title={},
    xlabel={Span},
    ylabel={Nr. sync. ops.},
    xmin=0, xmax=25,
    ymin=0, ymax=100000000,
    xtick={0,5,10,15,20,25},
    ytick={0,10,100,1000,10000,100000,1000000,10000000,100000000},
    legend pos=north west,
    ymajorgrids=true,
    grid style=dashed,
]

\addplot[
    color=blue,
    mark=+,
    ]
    coordinates {
    (0,64 + 64)
    (1,160 + 159)
    (2,288 + 291)
    (3,448 + 449)
    (4,592 + 581)
    (5,800 + 760)
    (6,1051 + 978)
    (7,1269 + 1092)
    (8,1383 + 1383)
    (9,2432 + 1745)
    (10,3456 + 1973)
    (11,5760 + 2442)
    (12,9920 + 2950)
    (13,18208 + 3514)
    (14,34742 + 4157)
    (15,67712 + 5059)
    (16,133416 + 5844)
    (17,264352 + 6659)
    (18,526771 + 7780)
    (19,1051209 + 8854)
    (20,2100004 + 10659)
    (21,4197147 + 11884)
    (22,8391530 + 13045)
    (23,16780309 + 15509)
    (24,33557664 + 16374)
    (25,67112296 + 19076)
};

\addplot[
    color=red,
    mark=x,
    ]
    coordinates {
    (0,64 + 64)
    (1,190 + 190)
    (2,352 + 352)
    (3,502 + 502)
    (4,785 + 785)
    (5,1037 + 1037)
    (6,1240 + 1240)
    (7,1522 + 1522)
    (8,1804 + 1804)
    (9,2059 + 2059)
    (10,2431 + 2431)
    (11,2688 + 2688)
    (12,2976 + 2976)
    (13,3334 + 3334)
    (14,3536 + 3536)
    (15,3876 + 3876)
    (16,4222 + 4222)
    (17,4776 + 4776)
    (18,5024 + 5024)
    (19,5363 + 5363)
    (20,5752 + 5752)
    (21,5828 + 5828)
    (22,6290 + 6290)
    (23,6703 + 6703)
    (24,6964 + 6964)
    (25,7427 + 7427)
};

\addplot[
    color=brown,
    mark=o,
    ]
    coordinates {
    (0,63)
    (1,188)
    (2,343)
    (3,521)
    (4,766)
    (5,1018)
    (6,1178)
    (7,1405)
    (8,1614)
    (9,1752)
    (10,1989)
    (11,2077)
    (12,2224)
    (13,2401)
    (14,2414)
    (15,2648)
    (16,2788)
    (17,3077)
    (18,3139)
    (19,3250)
    (20,3520)
    (21,3474)
    (22,3690)
    (23,3905)
    (24,3993)
    (25,4265)
};

\legend{CWS (CAS + MFence), LCWS (CAS + MFence), LCWS (Notifs)}

\end{semilogyaxis}
\end{tikzpicture}
}
\caption{64 procs, \#sync ops vs span, regular dags.} 
\label{fig:rec-nr-ops-span}
\end{subfigure}
\quad
\begin{subfigure}{0.32\textwidth}
\resizebox{\linewidth}{!}{
\begin{tikzpicture}
\begin{semilogyaxis}[
    title={},
    xlabel={Span},
    ylabel={Nr. sync. ops.},
    xmin=300, xmax=650,
    ymin=0, ymax=100000000,
    xtick={300,350,400,450,500,550,600,650},
    ytick={0,10,100,1000,10000,100000,1000000, 10000000, 100000000},
    legend pos=north west,
    ymajorgrids=true,
    grid style=dashed,
]

\addplot[
    color=blue,
    mark=+,
    ]
    coordinates {
    (300,27788 + 11977)
    (350,142188 + 12352)
    (400,535471 + 12882)
    (450,1062057 + 15729)
    (500,2113654 + 19289)
    (550,4211329 + 20013)
    (600,8408413 + 23112)
    (650,16794499 + 20928)
};

\addplot[
    color=red,
    mark=x,
    ]
    coordinates {
    (300,12380 + 12380)
    (350,13122 + 13122)
    (400,13848 + 13848)
    (450,15695 + 15695)
    (500,18593 + 18593)
    (550,20171 + 20171)
    (600,21548 + 21548)
    (650,19986 + 19986)
};
\addplot[
    color=brown,
    mark=o,
    ]
    coordinates {
    (300,11874)
    (350,12138)
    (400,12601)
    (450,14385)
    (500,17152)
    (550,18554)
    (600,19799)
    (650,18496)
};

\legend{CWS (CAS + MFence), LCWS (CAS + MFence), LCWS (Notifs)}
\end{semilogyaxis}
\end{tikzpicture}
}
\caption{64 procs, \#sync ops vs span, irregular dags.} 
\label{fig:unb-nr-ops-span}
\end{subfigure}
\begin{subfigure}{0.32\textwidth}
\resizebox{\linewidth}{!}{
\begin{tikzpicture}
\begin{semilogyaxis}[
    title={},
    xlabel={Processors executing the computation},
    ylabel={Nr. sync. ops.},
    xmin=1, xmax=64,
    ymin=0, ymax=10000000,
    xtick={1,8,16,24,32,40,48,56,64},
    ytick={0,10,100,1000,10000,100000,1000000, 10000000, 100000000},
    legend pos=south east,
    ymajorgrids=true,
    grid style=dashed,
]

\addplot[
    color=blue,
    mark=+,
    ]
    coordinates {
    (1,2097151 + 21)
    (2,2097156 + 45)
    (4,2097178 + 169)
    (8,2097384 + 1098)
    (16,2097744 + 2658)
    (32,2098418 + 5259)
    (64,2099814 + 10416)
    };

\addplot[
    color=red,
    mark=x,
    ]
    coordinates {
    (1,1 + 1 + 0)
    (2,12 + 12 + 8)
    (4,48 + 48 + 37)
    (8,431 + 431 + 268)
    (16,1298 + 1298 + 749)
    (32,2797 + 2797 + 1644)
    (64,5629 + 5629 + 3431)
    };

    \legend{CWS (CAS + MFence), LCWS (CAS + MFence + Notifs)}

\end{semilogyaxis}
\end{tikzpicture}
}
\caption{\#sync ops vs \#procs, regular dags.}
\label{fig:rec-nr-ops-procs}
\end{subfigure}
\caption{Comparison of Low Cost Work Stealing (LCWS) with Classical Work Stealing (CWS).}
\end{figure*}
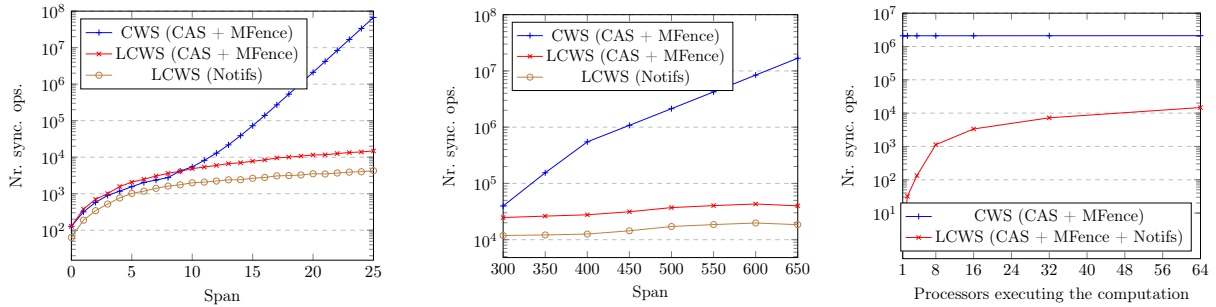

%% file: sections/5-conclusions-and-future-work.tex
\section{Conclusion}
\label{sec:conclusion}
In this paper we studied a \ws\ algorithm that uses \spds\ to reduce synchronization overheads.
Whereas traditional \ws\ algorithms require synchronization for every time processors access deques, in our proposal, synchronization operations are employed optimally, which is the key for eliminating most unnecessary synchronization overheads.
By default, busy processors operate locally on their deques without any synchronization, resembling a sequential execution.
Idle processors can request busy ones to expose some of their work, thus allowing for load balancing via direct steals.
This lazy approach for using synchronization is the key for guaranteeing an asymptotically optimal expected runtime while provably reducing synchronization overheads.
Indeed, we proved that the expected total synchronization of the algorithm is $O\left(PT_{\infty}\left(C_{\cas} + C_{\mfence}\right)\right)$.
To justify the tightness of our bounds, we recall that, for \cws, the expected number of (successful and unsuccessful) steal attempts is $O\left(PT_{\infty}\right)$.
By noting that the public part of an \spd\ is essentially a concurrent deque, and, by taking into account the impossibility of eliminating all synchronization from the implementation of concurrent deques while maintaining their correctness (see~\cite{DBLP:conf/popl/AttiyaGHKMV11}), we conclude that the synchronization bounds we have obtained for \cws\ are tight.


Finally, and as already discussed in Section~\ref{sub:comparison}, for several types of computations, the synchronization overheads of conventional \ws\ algorithms grow linearly with both the total amount of work and the number of steal attempts.
%
For numerous classes of parallel computations, the total amount of work increases exponentially with the span of the computation (i.e. $T_{1} = O\left(2^{T_{\infty}}\right)$).
From this perspective, our results make evident the significance of the improvement of \cws\ over prior \ws\ algorithms: not only are the synchronization overheads incurred by our algorithm (essentially) exponentially smaller than previous algorithms, but our algorithm also maintains the asymptotically optimal expected runtime bounds of the concurrent deque \ws\ algorithm~\cite{DBLP:journals/mst/AroraBP01}.
%




%% file: sections/appendix-1.tex
\section*{Appendix: Proofs}
\label{sec:proofs}
Due to space constraints, most of the proofs of our claims are placed in this appendix.
Nevertheless, if this paper is accepted, we will publish a full version of this paper, with proofs, in a freely accessible on-line repository.



The following lemma is crucial for the performance analysis of \cws.
An analogous result has already been proved for concurrent deques (see~\cite[Lemma 3]{DBLP:journals/mst/AroraBP01}).
For the sake of completion we present its proof, which is a simple transcription of original proof of~\cite[Lemma 3]{DBLP:journals/mst/AroraBP01}, adapted for \spds.

\begin{lemma}[Structural Lemma for \spds]
\label{lemma:app-structural lemma}
Let $v_{1},\ldots,v_{k}$ denote the nodes stored in some processor $p$'s \spd, ordered from the bottom of the \spd\ to the top, at some point in the linearized execution of \cws.
Moreover, let $v_{0}$ denote $p$'s assigned node (if any), and for $i = 0,\ldots,k$ let $u_{i}$ denote the designated parent of $v_{i}$ in the enabling tree.
Then, for $i = 1,\ldots,k$, $u_{i}$ is an ancestor of $u_{i-1}$ in the enabling tree, and despite $v_{0}$ and $v_{1}$ may have the same designated parent (\textit{i.e.}~$u_{0} = u_{1}$), for $i = 2,3,\ldots,k$, $u_{i-1} \neq u_{i}$ (\textit{i.e.}~the ancestor relationship is proper).
\end{lemma}
\begin{proof}
Fix a particular \spd.
The \spd\ state and assigned node only change when the assigned node is executed or a thief performs a successful steal.
We prove the claim by induction on the number of assigned-node executions and steals since the \spd\ was last empty.
In the base case, if the \spd\ is empty, then the claim holds vacuously.
We now assume that the claim holds before a given assigned-node execution or successful steal, and we will show that it holds after.
Specifically, before the assigned-node execution or successful steal, let $v_{0}$ denote the assigned node;
let $k$ denote the number of nodes in the \spd;
let $v_{1},\ldots,v_{k}$ denote the nodes in the \spd\ ordered from the bottom to top;
and for $i=0,\ldots,k$, let $u_{i}$ denote the designated parent of $v_{i}$.
We assume that either $k = 0$, or for $i = 1,\ldots,k$, node $u_{i}$ is an ancestor of $u_{i-1}$ in the enabling tree, with the ancestor relationship being proper, except possibly for the case $i = 1$.
After the assigned-node execution or successful steal, let ${v_{0}}'$ denote the assigned node;
let ${k}'$ denote the number of nodes in the \spd;
let ${v_{1}}', \ldots, {v_{k}}'$ denote the nodes in the \spd\ ordered from bottom to top;
and for $i = 1,\ldots,{k}'$, let ${u_{i}}'$ denote the designated parent of ${v_{i}}'$.
We now show that either ${k}' = 0$, or for $i = 1,\ldots,{k}'$, node ${u_{i}}'$ is an ancestor of ${u_{i-1}}'$ in the enabling tree, with the ancestor relationship being proper, except possibly for the case $i = 1$.

Consider the execution of the assigned node $v_{0}$ by the owner.

If the execution of $v_{0}$ enables 0 children, then the owner pops the bottommost node off its \spd\ and makes that node its new assigned node.
If $k = 0$, then the \spd\ is empty;
the owner does not get a new assigned node;
and ${k}' = 0$.
If $k > 0$, then the bottommost node $v_{1}$ is popped and becomes the new assigned node, and ${k}' = k - 1$.
If $k = 1$, then ${k}' = 0$.
Otherwise, ${k}' = k - 1$.
We now rename the nodes as follows.
For $i = 0,\ldots,{k}'$, we set ${v_{i}}' = v_{i + 1}$ and ${u_{i}}' = u_{i + 1}$.
We now observe that for $i = 1,\ldots,{k}'$, node ${u_{i}}'$ is a proper ancestor of ${u_{i - 1}}'$ in the enabling tree.

If the execution of $v_{0}$ enables 1 child $x$, then $x$ becomes the new assigned node;
the designated parent of $x$ is $v_{0}$;
and ${k}' = k$.
If $k = 0$, then ${k}' = 0$.
Otherwise, we can rename the nodes as follows.
We set ${v_{0}}' = x$;
we set ${u_{0}}' = v_{0}$;
and for $i = 1,\ldots,{k}'$, we set ${v_{i}}' = v_{i}$ and ${u_{i}}' = u_{i}$.
We now observe that for $i = 1,\ldots,{k}'$, node ${u_{i}}'$ is a proper ancestor of ${u_{i - 1}}'$ in the enabling tree.
That ${u_{1}}'$ is a proper ancestor of ${u_{0}}'$ in the enabling tree follows from the fact that $\left(u_{0},v_{0}\right)$ is an enabling edge.

In the most interesting case, the execution of the assigned node $v_{0}$ enables 2 children $x$ and $y$, with $x$ being pushed onto the bottom of the \spd\ and $y$ becoming the new assigned node.
In this case, $\left(v_{0},x\right)$ and $\left(v_{0},y\right)$ are both enabling edges, and ${k}' = k + 1$.
We now rename the nodes as follows.
We set ${v_{0}}' = y$;
we set ${u_{0}}' = v_{0}$;
we set ${v_{1}}' = x$;
we set ${u_{1}}' = v_{0}$;
and for $i = 2,\ldots,{k}'$, we set ${v_{i}}' = v_{i - 1}$ and ${u_{i}}' = u_{i - 1}$.
We now observe that ${u_{1}}' = {u_{0}}'$, and for $i = 2,\ldots,{k}'$, node ${u_{i}}'$ is a proper ancestor of ${u_{i - 1}}'$ in the enabling tree.
That ${u_{2}}'$ is a proper ancestor of ${u_{1}}'$ in the enabling tree follows from the fact that $\left(u_{0},v_{0}\right)$ is an enabling edge.

Finally, we consider a successful steal by a thief.
In this case, the thief pops the topmost node $v_{k}$ off the \spd, so ${k}' = k - 1$.
If $k = 1$, then ${k}' = 0$.
Otherwise, we can rename the nodes as follows.
For $i = 0,\ldots,{k}'$, we set ${v_{i}}' = v_{i}$ and ${u_{i}}' = u_{i}$.
We now observe that for $i = 1,\ldots,{k}'$, node ${u_{i}}'$ is an ancestor of ${u_{i-1}}'$ in the enabling tree, with the ancestor relationship being proper, except possibly for the case $i = 1$.
\end{proof}

\begin{corollary}
\label{corollary:app-structural corollary}
If $v_{0},\,v_{1},\,\ldots,\,v_{k}$ are as defined in the statement of Lemma~\ref{lemma:app-structural lemma}, then we have $w\left(v_{0}\right) \leq w\left(v_{1}\right) < \ldots < w\left(v_{k-1}\right) < w\left(v_{k}\right)$.
\end{corollary}

We are now able to bound the execution time of a computation depending on the number of idle iterations that take place during that computation's execution.
The following result is a trivial variant of~\cite[Lemma 5]{DBLP:journals/mst/AroraBP01} but considering the \cws\ algorithm, and is only added for the sake of completion.

\begin{lemma}
\label{lemma:app-runtime-by-iterations}
Consider any computation with work $T_{1}$ being executed by $P$ processors, under \cws.
The execution time is $O\left(\frac{T_{1}}{P} + \frac{I}{P}\right)$, where $I$ denotes the number of idle iterations executed by processors.
\end{lemma}
\begin{proof}
Consider two buckets to which we add tokens during the computation's execution: the \textit{busy} bucket and the \textit{idle} bucket.
At the end of each iteration, every processor places a token into one of these buckets.
If a processor executed a node during the iteration, it places a token into the busy bucket, and otherwise, it places a token into the idle bucket.
Since we have $P$ processors, for each $C$ consecutive steps, at least $P$ tokens are placed into the buckets.

Because, by definition, the computation has $T_{1}$ nodes, there will be exactly $T_{1}$ tokens in the busy bucket when the computation's execution ends.
Moreover, as $I$ denotes the number of idle iterations, it also corresponds to the number of tokens in the idle bucket when the computation's execution ends.
Thus, exactly $T_{1} + I$ tokens are collected during the computation's execution.
Taking into account that for each $C$ consecutive steps at least $P$ tokens are placed into the buckets, we conclude the number of steps required to collect all the tokens is at most $C . \left(\frac{T_{1}}{P} + \frac{I}{P}\right)$.
After collecting all the $T_{1}$ tokens, the computation's execution terminates, implying the execution time is at most $O\left(\frac{T_{1}}{P} + \frac{I}{P}\right)$.
\end{proof}


\begin{lemma}
\label{lemma:app-syncfree-iteration}
Consider a processor $p$ executing a busy iteration such that $p$'s $targeted$ flag is set to \textsc{false} when $p$ checks it at the beginning of the iteration.
If the execution of $p$'s assigned node enables one or more nodes, or, if the private part of $p$'s \spd\ is not empty, then, no \mfence\ instruction is required during the execution of the iteration.
\end{lemma}
\begin{proof}
Consider the \cws\ algorithm, depicted in Algorithm~\ref{algo:cws}.
The first action $p$ takes for the execution of that iteration is checking the value of its $targeted$ flag (line 3).
Since, by the statement of this lemma $p$'s $targeted$ flag is set to \textsc{false} at the moment when $p$ checks the flag's value, $p$ does not enter the \textsl{then} branch of the \textsf{if} statement.
Moreover, as a consequence of the conditional statement of line 3, there is a control dependency that does not allow the instructions succeeding the conditional expression to be reordered with the evaluation of the condition, implying no \mfence\ instruction is required until this point.

%
After that, $p$ checks if it has an assigned node.
Again, since the next action $p$ takes depends on its currently assigned node, there is a control dependency from the instruction where $p$ checks if it has a currently assigned node to both branches of the if statement.
Thus, no instruction reordering between the evaluation of the condition and any of the instructions succeeding that evaluation can be made, implying no \mfence\ instruction is required until here.

Because we assumed $p$ was executing a busy iteration, by the definition of a busy iteration, $p$ must have an assigned node.
Hence, $p$ executes its assigned node.
Since the next action $p$ takes depends on the outcome of that node's execution, there is a control dependency between the execution of $p$'s assigned node and the execution of the sequence of instructions corresponding to each of the possible outcomes.
Hence, no instruction reordering can be made, implying no \mfence\ instruction is necessary until this point.

From that node's execution, three outcomes are possible:
\begin{description}
  \item[0 nodes enabled] In this case, $p$ invokes the \textsl{pop} method to its own \spd.
    By Lemma~\ref{lemma:syncfree-pop}, the invocation does not require the execution of a \mfence\ instruction.
    Furthermore, the next instruction (line 16) has a data dependency on the value of $p$'s assigned node, for which reason it cannot be reordered with the invocation of the \textsl{pop} method and so no \mfence\ instruction is required.

    Since we have assumed that the private part of $p$'s \spd\ was not empty, it is trivial to conclude that the \textsl{pop} invocation returns a node, which immediately becomes $p$'s new assigned node.
    Thus, after having a new node assigned $p$ takes no further action during the iteration, meaning no \mfence\ instruction was required for the execution of the iteration in this situation.

  \item[1 node enabled] In this case the enabled node becomes $p$'s new assigned node.
    Next, $p$ checks the number of nodes that were enabled.
    The assignment of one of the enabled nodes and the instruction where $p$ checks the number of nodes enabled can be reordered.
    Fortunately, because $enabled$ is a local variable that is solely accessed by $p$, there is no harm for a parallel execution if the instructions are reordered and so no \mfence\ instruction is required for this case as well.
    Because $p$ enabled a single node it takes no further action during the iteration, implying the lemma holds in this situation as well.

  \item[2 nodes enabled] Finally, for this case one of the enabled nodes becomes $p$'s new assigned node.
    Using the same reasoning as for the case where a single node was enabled, we conclude that no \mfence\ instruction is required at least until the evaluation of the conditional statement.
    Because $p$ enabled two nodes, $p$ enters the conditional expression and pushes the node it did not assign into the bottom of its \spd, by invoking the \textsl{push} method.
    Since there is a control dependency between the execution of this instruction and the evaluation of the condition, no instruction reordering is allowed.
    Thus, no \mfence\ instruction is required between these two instructions.

    Finally, Lemma~\ref{lemma:syncfree-push} states that an invocation to the \textsl{push} method does not require a \mfence\ instruction to be executed.
    Because after the invocation $p$ takes no further action during the iteration, we deduce the lemma holds, concluding its proof.
\end{description}
\end{proof}

The following lemma is a consequence of Lemma~\ref{lemma:syncfree-iteration} (corresponding to Lemma~\ref{lemma:app-syncfree-iteration} of the appendix) and states that the number of \cas\ and \mfence\ instructions executed during a computation's execution using \cws\ only depends on the number of idle iterations and processors.
\begin{lemma}
\label{lemma:app-cas-and-mem-barriers-by-throws}
Consider any computation being executed by the \cws\ algorithm, using $P$ processors.
The number of \cas\ and \mfence\ instructions executed by processors during the computation's execution is at most $O\left(I + P\right)$, where $I$ denotes the total number of idle iterations executed by processors.
\end{lemma}
\begin{proof}



By observing Algorithms~\ref{algo:cws} and~\ref{algo:spdeque}, it is easy to see that only invocations to \textsl{popBottom} or \textsl{popTop} methods can lead to the execution of \cas\ instructions.
Furthermore, both these methods are invoked at most once per scheduling iteration, and, for both, at most one \cas\ instruction is executed per invocation.
Since processors only invoke the \textsl{popTop} method when executing idle iterations, the number of \cas\ instructions caused by invocations to \textsl{popTop} is $O\left(I\right)$.
On the other hand, processors only invoke the \textsl{popBottom} method during busy iterations where the private part of their \spd\ is empty and the execution of their currently assigned node does not enable any new nodes.
Let $p$ denote some processor executing one such iteration.
From $p$'s invocation to the \textsl{popBottom} method two outcomes are possible:
\begin{description}
  \item[A node is returned] In this case the public part of $p$'s \spd\ was not empty implying $p$ had previously transferred a node from the private part of its \spd\ to the public part.
    By observing Algorithm~\ref{algo:cws} it is easy to deduce that $p$ only makes these node transfers if some thief had previously set $p$'s $targeted$ flag to \textsc{true}.
    Moreover, because after transferring the node $p$ immediately sets its $targeted$ flag back to \textsc{false}, the number of times $p$ makes such node transfers is at most the number of times it is targeted by a steal attempt.
    Taking into account that processors only make steal attempts during the execution of idle iterations, and make exactly one steal attempt for each such iteration, exactly $I$ steal attempts take place during a computation's execution.
    As such, the number of \cas\ instructions executed in situations like this one is at most $O\left(I\right)$.

  \item[\textsc{Empty} is returned] In this case $p$ will not have an assigned node at the end of the scheduling iteration's execution.
    Thus, after $p$ finishes executing the iteration two scenarios may occur:
    \begin{description}
      \item[$p$ executes an idle iteration] For this case we can create a mapping from idle iterations to each busy iteration that precedes an idle iteration, implying there can be at most $O\left(I\right)$ such iterations.
        With this, it is trivial to conclude that the number of \cas\ instructions executed by \cws\ for situations equivalent to this one is at most $O\left(I\right)$.
      \item[The computation's execution terminates] Because there are exactly $P$ processors, at most $P$ scheduling iterations can precede the end of a computation's execution.
        Consequently, the number of \cas\ instructions executed for scenarios equivalent to this one is at most $O\left(P\right)$.
    \end{description}
\end{description}

Summing up all the possible scenarios, the number of \cas\ instructions executed by \cws\ is at most $O\left(I + P\right)$.


We now turn to the number of \mfence\ instructions executed during a computation's execution.
To that end, we first bound the number of scheduling iterations that can contain \mfence\ instructions.
Consider any scheduling iteration $s$ during a computation's execution, and let $p$ denote the processor that executed the iteration.
Iteration $s$ was either an idle or a busy iteration.
\begin{description}
  \item[$s$ is an idle iteration] By definition, at most $I$ iterations are idle, implying there are $O\left(I\right)$ such iterations that could contain \mfence\ instructions.

  \item[$s$ is a  busy iteration] When $p$ checks its $targeted$ flag, one of the two following situations arises:
      \begin{description}
        \item[$targeted$ is \textsc{true}] By observing Algorithm~\ref{algo:cws} we conclude that such a situation can only occur if another processor $q$ has set $p$'s $targeted$ to \textsc{true}, which can only occur if $q$ was executing an idle iteration.
          After executing the conditional statement, $p$ resets its $targeted$ flag back to \textsc{false}.
          Thus, the total number of busy iterations where a processor has its flag set to $targeted$ is at most $I$, because each such iteration can be mapped by an idle iteration.
          Consequently, the number of iterations similar to this one is at most $O\left(I\right)$.
        \item[$targeted$ is \textsc{false}] As $p$ is executing a busy iteration, it will execute the node it has assigned.
          From that node's execution, either 0, 1 or 2 other nodes can be enabled.
            \begin{description}
              \item[More than 0 nodes are enabled] Lemma~\ref{lemma:syncfree-iteration} (corresponding to Lemma~\ref{lemma:app-syncfree-iteration} of the appendix) implies that no \mfence\ instruction is executed in this case.
              \item[0 nodes are enabled] In this case, $p$ cannot immediately assign a new node, because it did not enable any.
                By Algorithm~\ref{algo:cws}, $p$ will then invoke the \textsl{pop} method to its own \spd.
                With this, one of two possible situations arises:
                  \begin{description}
                    \item[\spd's private part is not empty] As a consequence of Lemma~\ref{lemma:syncfree-iteration} (corresponding to Lemma~\ref{lemma:app-syncfree-iteration} of the appendix), no \mfence\ instruction is executed in this case.
                    \item[\spd's private part is empty] In this case, by observing Algorithm~\ref{algo:spdeque} we conclude that the invocation returns the special value \textsc{race}, implying $p$ will make an invocation to \textsl{popBottom} still during that same iteration.
                      From that invocation, two outcomes are possible:
                        \begin{description}
                          \item[A node is returned] In this situation, $p$ assigns the node.
                            By observing Algorithm~\ref{algo:cws} it is trivial to conclude that this scenario only arises if some processor previously set $p$'s $targeted$ flag to \textsc{true}.
                            As a consequence, $p$ transfered a node from the private part of its \spd\ to the public part.
                            Again, using the same reasoning as for the case where $p$'s $targeted$ flag is set to \textsc{true}, we conclude the number of such iterations is at most $O\left(I\right)$.
                          \item[\textsc{empty} is returned] After $p$ finishes executing the current scheduling iteration $s$, two scenarios may occur:
                              \begin{description}
                                \item[$p$ executes an idle iteration] It is trivial to deduce that we can create a mapping from idle iterations to each iteration satisfying the same conditions as $s$.
                                  Thus, there can be at most $O\left(I\right)$ such iterations.
                                \item[The computation's execution terminates] Since there are exactly $P$ processors, at most $P$ scheduling iterations can precede the end of a computation's execution.
                                  Consequently, there are at most $P$ scheduling iteration similar to $s$.
                              \end{description}
                        \end{description}
                  \end{description}
            \end{description}
      \end{description}
\end{description}

Now, we sum up all the scheduling iterations that may contain \mfence\ instructions.
Accounting with all possible scenarios it follows that at most $O\left(I + P\right)$ scheduling iterations may contain \mfence\ instructions.
Since any scheduling iteration is composed by at most $C$ instructions, at most $C$ \mfence\ instructions can be executed per iteration, implying the number of \mfence\ instructions executed during a computation's execution is at most $O\left(I + P\right)$.
\end{proof}

The following lemma is a formalization of the arguments already given in~\cite{DBLP:journals/mst/AroraBP01}, but considering the potential function we present.

\begin{lemma}
\label{lemma:app-potential properties}
Consider some node $u$, ready at step $i$ during the execution of a computation.
\begin{enumerate}
  \item If $u$ gets assigned to a processor at that step, the potential drops by at least $\frac{3}{4}\phi_{i}\left(u\right)$.
  \item If $u$ becomes stealable at that step, the potential drops by at least $\frac{3}{4}\phi_{i}\left(u\right)$.
  \item If $u$ was already assigned to a processor and gets executed at that step $i$, the potential drops by at least $\frac{47}{64}\phi_{i}\left(u\right)$.
\end{enumerate}
\end{lemma}
\begin{proof}
Regarding the first claim, if $u$ was stealable the potential decreases from $4^{3w\left(u\right) - 1}$ to $4^{3w\left(u\right)-2}$.
Otherwise, the potential decreases from $4^{3w\left(u\right)}$ to $4^{3w\left(u\right)-2}$, which is even more than in the previous case.
Given that $4^{3w\left(u\right) - 1} - 4^{3w\left(u\right)-2} = \frac{3}{4}\phi_{i}\left(u\right)$, we conclude that if $u$ gets assigned the potential decreases by at least $\frac{3}{4}\phi_{i}\left(u\right)$.

Regarding the second one, note that $u$ was not stealable (because it became stealable at step $i$) and so the potential decreases from $4^{3w\left(u\right)}$ to $4^{3w\left(u\right)-1}$.
So, if $u$ becomes stealable, the potential decreases by $4^{3w\left(u\right)} - 4^{3w\left(u\right)-1} = \frac{3}{4}\phi_{i}\left(u\right)$.

We now prove the last claim.
Remind that, by our conventions regarding computations' structure, each node within a computation's dag can have an out-degree of at most two.
Consequently, each node can be the designated parent of at most two other ones in the enabling tree.
Moreover, by definition, the weight of any node is strictly smaller than the weight of its designated parent, since it is deeper in the enabling tree than its designated parent.
Consider the three possible scenarios:
\begin{description}
  \item[0 nodes enabled] The potential decreased by $\phi_{i}\left(u\right)$.
  \item[1 node enabled] The enabled node becomes the assigned node of the processor (that executed $u$).
    Let $x$ denote the enabled node. Since $x$ is the child of $u$ in the enabling tree, it follows $\phi_{i}\left(u\right) - \phi_{i+1}\left(x\right) = 4^{3w\left(u\right) - 2} - 4^{3w\left(x\right) - 2} = 4^{3w\left(u\right) - 2} - 4^{3\left(w\left(u\right) - 1 \right) - 2} = \frac{63}{64}\phi_{i}\left(u\right)$.
    Thus, for this situation, the potential decreases by $\frac{63}{64}\phi_{i}\left(u\right)$.
  \item[2 nodes enabled] In this case, one of the enabled nodes immediately becomes the assigned node of the processor whist the other is pushed onto the bottom of the \spd's private part.
    Let $x$ denote the enabled node that becomes the processor's new assigned node and $y$ the other enabled node.
    Since both $x$ and $y$ have $u$ as their designated parent in the enabling tree, we have $\phi_{i}\left(u\right) - \phi_{i+1}\left(x\right) - \phi_{i+1}\left(y\right) = 4^{3w\left(u\right) - 2} - 4^{3w\left(x\right)} - 4^{3w\left(y\right) - 2} = \frac{47}{64}\phi_{i}\left(u\right)$.
    As such, the potential decreases by $\frac{47}{64}\phi_{i}\left(u\right)$, concluding the proof of the lemma.
\end{description}
\end{proof}

The following lemma is a direct consequence of Corollary~\ref{corollary:structural corollary} (corresponding to Corollary~\ref{corollary:app-structural corollary} of the appendix) and of the potential function's properties.
The result is a variant of~\cite[Top-Heavy Deques]{DBLP:journals/mst/AroraBP01}, considering \spds~instead of the conventional fully concurrent deques, and our potential function, instead of the original.

\begin{lemma}
\label{lemma:app-top-heavy spdeques}
Consider any step $i$ and any processor $p \in D_{i}$.
The top-most node $u$ in $p$'s \spd\ contributes at least $\frac{4}{5}$ of the potential associated with $p$.
That is, we have $\phi_{i}\left(u\right) \geq \frac{4}{5}\Phi_{i}\left(p\right)$.
\end{lemma}
\begin{proof}
This lemma follows from Corollary~\ref{corollary:structural corollary} (corresponding to Corollary~\ref{corollary:app-structural corollary} of the appendix).
We prove it by induction on the number of nodes within $p$'s \spd.
\begin{description}
  \item[Base case] As the base case, consider that $p$'s \spd\ contains a single node $u$.
  The processor itself can either have or not an assigned node.
  For the second scenario, we have $\phi_{i}\left(u\right) = \Phi_{i}\left(p\right)$.
  Regarding the first case, let $x$ denote $p$'s assigned node.
  Corollary~\ref{corollary:structural corollary} implies that $w\left(u\right) \geq w\left(x\right)$.
  It follows $\Phi_{i}\left(q\right) = \phi_{i}\left(u\right) + \phi_{i}\left(x\right) = 4^{3w\left(u\right) - 1} + 4^{3w\left(x\right) - 2} \leq \frac{5}{4}\phi_{i}\left(u\right)$.
  Thus, if $p$'s \spd\ contains a single node we have $\Phi_{i}\left(q\right) \leq \frac{5}{4}\phi_{i}\left(u\right)$.

  \item[Induction step] Consider that $p$'s \spd\ now contains $n$ nodes, where $n \geq 2$, and let $u,\,x$ denote the topmost and second topmost nodes, respectively, within the \spd.
  For the purpose of induction, let us assume the lemma holds for all the first $n - 1$ nodes (\textit{i.e.}~without accounting with $u$):
  $\Phi_{i}\left(q\right) - \phi_{i}\left(u\right) \leq \frac{5}{4}\phi_{i}\left(x\right)$.
  Corollary~\ref{corollary:structural corollary} (corresponding to Corollary~\ref{corollary:app-structural corollary} of the appendix) implies $w\left(u\right) > w\left(x\right) \equiv w\left(u\right) - 1 \geq w\left(x\right)$.
  It follows
  $\Phi_{i}\left(q\right) \leq \frac{5}{4}\phi_{i}\left(x\right) + \phi_{i}\left(u\right) = \frac{5}{4}4^{3w\left(x\right)} + 4^{3w\left(u\right)} \leq \frac{5}{4}4^{3\left(w\left(u\right) - 1\right)} + 4^{3w\left(u\right)} = \frac{261}{256}\phi_{i}\left(u\right) < \frac{5}{4}\phi_{i}\left(u\right)$
  concluding the proof of the lemma.
\end{description}
\end{proof}

The following result is a consequence of Lemma~\ref{lemma:top-heavy spdeques} (corresponding to Lemma~\ref{lemma:app-top-heavy spdeques} of the appendix).

\begin{lemma}
\label{lemma:app-potential-decrease}
Suppose a thief processor $p$ chooses a processor $q \in D_{i}$ as its victim at some step $j$, such that $j \geq i$\ (\textit{i.e.}~a steal attempt of $p$ targeting $q$ occurs at step $j$).
Then, at step $j + 2C$, the potential decreased by at least $\frac{3}{5}\Phi_{i}\left(q\right)$ due to either the assignment of the topmost node in $q$'s \spd, or for making the topmost node of $q$'s \spd\ become stealable.
\end{lemma}
\begin{proof}
Let $u$ denote the topmost node of $q$'s \spd\ at the beginning of step $i$.
We first prove that $u$ either gets assigned or becomes stealable.

Three possible scenarios may take place due to $p$'s steal attempt targeting $q$'s \spd.
\begin{description}
  \item[The invocation returns a node] If $p$ stole $u$, then, $u$ gets assigned to $p$.
    Otherwise, some other processor removed $u$ before $p$ did, implying $u$ got assigned to that other processor.

  \item[The invocation aborts] Since the \spd\ implementation meets the \spdrs~on any good set of invocations, and because the \cws\ algorithm only makes good sets of invocations, we conclude that some other processor successfully removed a topmost node from $q$'s \spd\ during the aborted steal attempt made by $p$.
  If the removed node was $u$, $u$ gets assigned to a processor (that may either be $q$, or, some other thief that successfully stole $u$).
  Otherwise, $u$ must have been previously stolen by a thief or popped by $q$, and thus became assigned to some processor.

  \item[The invocation returns \textsc{empty}] This situation can only occur if either $q$'s \spd\ is completely empty, or if there is no node in the public part of $q$'s \spd.
    \begin{itemize}
      \item For the first case, since $q \in D_{i}$, some processor must have successfully removed $u$ from $q$'s \spd.
        Consequently, $u$ was assigned to a processor.

      \item If there was no node in the public part of $q$'s \spd, $p$ sets $q$'s $targeted$ flag to \textsc{true} in a later step ${j}'$.
        Recall that, for each $C$ consecutive instructions executed by a processor, at least one corresponds to a milestone.
        It follows that ${j}' \leq j + C$.
        Furthermore, by observing Algorithm~\ref{algo:cws}, we conclude that $q$ will make and complete an invocation to \textsl{updateBottom} of its \spd\ in one of the $C$ steps succeeding step ${j}'$.
        Thus if $q$'s \spd's private part is not empty, a node will become stealable.
        From that invocation, only two possible situations can take place:
          \begin{description}
            \item[No node becomes stealable] In this case, the private part of $q$'s \spd\ was empty, implying some processor (either $q$ or some thief) assigned $u$.

            \item[A node becomes stealable] If the node that became stealable as the result of the invocation was not $u$, then either $u$ was assigned by a processor (that could have been $q$ or some thief), or $u$ had already been transfered to the public part of $q$'s \spd\ as a consequence of another thief's steal attempt that also returned \textsc{empty}, implying that either $u$ became assigned, or it became stealable.
              Otherwise, the node that became stealable as a result of the \textsl{updateBottom}'s invocation was $u$.
              Thus, in any case, $u$ either gets assigned to a processor or becomes stealable.
          \end{description}
    \end{itemize}
\end{description}
With this, we conclude that $u$ either became assigned or became stealable until step $j + 2C$.

From Lemma~\ref{lemma:top-heavy spdeques} (corresponding to Lemma~\ref{lemma:app-top-heavy spdeques} of the appendix), we have $\phi_{i}\left(u\right) \geq \frac{4}{5}\Phi_{i}\left(q\right)$.
Furthermore, Lemma~\ref{lemma:potential properties} (corresponding to Lemma~\ref{lemma:app-potential properties} of the appendix) proves that if $u$ gets assigned the potential decreases by at least $\frac{3}{4}\phi_{i}\left(u\right)$, and if $u$ becomes stealable the potential also decreases by at least $\frac{3}{4}\phi_{i}\left(u\right)$.
Because $u$ is either assigned or becomes stealable in any case, we conclude the potential associated with $q$ at step $j + 2C$ has decreased by at least $\frac{3}{5}\Phi_{i}\left(q\right)$.
\end{proof}


The following lemma is trivial a generalization of the original result presented in~\cite[Balls and Weighted Bins]{DBLP:journals/mst/AroraBP01}.
The only difference between the two results is the assumption of having at least $B$ balls, rather than exactly $B$ balls.
Its proof is only presented for the sake of completion, and is (trivially) adapted from the proof of~\cite[Balls and Weighted Bins]{DBLP:journals/mst/AroraBP01}.

\begin{lemma}[Balls and Weighted Bins]
\label{lemma:app-balls and weighted bins}
Suppose we are given at least $B$ balls, and exactly $B$ bins.
Each of the balls is tossed independently and uniformly at random into one of the $B$ bins, where for $i = 1,\ldots,\,B$, bin $i$ has a weight $W_{i}$.
The total weight is $W = \sum_{i = 1}^{B} W_{i}$.
For each bin $i$, we define the random variable $X_{i}$ as \[X_{i} = \left\{\begin{matrix}
W_{i} & \text{if some ball lands in bin } i\\
0 & \text{otherwise}
\end{matrix}\right.\]
and define the random variable $X$ as $X = \sum_{i = 1}^{B} X_{i}$.

Then, for any $\beta$ in the range $0 < \beta < 1$, we have $P\left\{X \geq \beta W\right\} \geq 1 - \frac{1}{\left(1 - \beta\right)\euler}$.
\end{lemma}
\begin{proof}
Consider the random variable $W_{i} - X_{i}$ taking the value of $W_{i}$ when no ball lands in bin $i$ and $0$ otherwise, and let ${B}'$ denote the total number of balls that are tossed.
It follows
$E\left[W_{i} - X_{i}\right] = W_{i}\left(1 - \frac{1}{B}\right)^{{B}'} \leq \frac{W_{i}}{\euler}$.
From the linearity of expectation, we have $E\left[W-X\right] \leq \frac{W}{\euler}$.
Markov's Inequality then implies
$P\left\{W - X > \left(1 - \beta \right )W\right\} = P\left\{X < \beta W\right\} \leq \frac{E\left[W - X\right]}{\left(1 - \beta \right )W} \leq \frac{1}{\left(1 - \beta \right )\euler}$.
%
\end{proof}


The following result states that for each $P$ idle iterations that take place, with constant probability the potential drops by a constant factor.
An analogous lemma was originally presented in~\cite[Lemma 8]{DBLP:journals/mst/AroraBP01} for the non-blocking Work Stealing algorithm.
The result is a consequence of Lemmas~\ref{lemma:balls and weighted bins} and~\ref{lemma:potential-decrease} (corresponding to Lemmas~\ref{lemma:app-balls and weighted bins} and~\ref{lemma:app-potential-decrease} of the appendix, respectively) and its proof follows the same traits as the one presented in that study.

\begin{lemma}
\label{lemma:app-phase potential decrease}
Consider any step $i$ and any later step $j$ such that at least $P$ idle iterations occur from $i$\ (inclusive) to $j$\ (exclusive).
Then, we have \[P\left\{\Phi_{i} - \Phi_{j + 2C} \geq \frac{3}{10}\Phi_{i}\left(D_{i}\right)\right\} > \frac{1}{4}.\]
\end{lemma}
\begin{proof}
By Lemma~\ref{lemma:potential-decrease} (corresponding to Lemma~\ref{lemma:app-potential-decrease} of the appendix) we know that for each processor $p \in D_{i}$ that is targeted by a steal attempt, the potential drops by at least $\frac{3}{5}\Phi_{i}\left(p\right)$, at most $2C$ steps after being targeted.

When executing an idle iteration, a processor plays the role of a thief attempting to steal work from some victim.
Thus, since $P$ idle iterations occur from step $i$\ (inclusive) to step $j$\ (exclusive), at least $P$ steal attempts take place during that same interval.
We can think of each such steal attempt as a ball toss of Lemma~\ref{lemma:balls and weighted bins} (corresponding to Lemma~\ref{lemma:app-balls and weighted bins} of the appendix).

For each processor $p$ in $D_{i}$, we assign it a weight $W_{p} = \frac{3}{5}\Phi_{i}\left(p\right)$, and for each other processor $p$ in $A_{i}$, we assign it a weight $W_{p} = 0$.
Clearly, the weights sum to $W = \frac{3}{5}\Phi_{i}\left(D_{i}\right)$.
Using $\beta = \frac{1}{2}$ in Lemma~\ref{lemma:balls and weighted bins} (Lemma~\ref{lemma:app-balls and weighted bins} of the appendix) it follows that with probability at least $1 - \frac{1}{\left(1 - \beta\right)\euler} > \frac{1}{4}$,
the potential decreases by at least
$\beta W = \frac{3}{10}\Phi_{i}\left(D_{i}\right)$,
concluding the proof of this lemma.
\end{proof}


Finally, we bound the expected number of idle iterations that take place during a computation's execution using the \cws\ algorithm.
The result follows from Lemma~\ref{lemma:phase potential decrease} (corresponding to Lemma~\ref{lemma:app-phase potential decrease} of the appendix) and is proved using similar arguments as the ones used in the proof of~\cite[Theorem 9]{DBLP:journals/mst/AroraBP01}.
The presented proof corresponds to an adaptation of the one originally presented for the just mentioned Theorem.

\begin{lemma}
\label{lemma:app-bounded-idle-iterations}
Consider any computation with work $T_{1}$ and critical-path length $T_{\infty}$ being executed by \cws\ using $P$ processors.
The expected number of idle iterations is at most $O\left(P T_{\infty}\right)$.
Moreover, with probability at least $1 - \varepsilon$ the number of idle iterations is at most $O\left(\left(T_{\infty} + \ln\left(\frac{1}{\varepsilon}\right)\right)P\right)$.
\end{lemma}
\begin{proof}
To analyze the number of idle iterations, we break the execution into \emph{phases}, each composed by $\Theta\left(P\right)$ idle iterations.
Then, we prove that, with constant probability, a phase leads the potential to drop by a constant factor.

A computation's execution begins when the root gets assigned to a processor.
By definition, the weight of the root is $T_{\infty}$, implying the potential at the beginning of a computation's execution starts at $\Phi_{0} = 4^{3T_{\infty} - 2}$.
Furthermore, it is straightforward to deduce that the potential is $0$ after (and only after) a computation's execution terminates.
We use these facts to bound the expected number of phases needed to decrease the potential down to $0$.
The first phase starts at step $t_{1} = 1$, and ends at the first step ${t_{1}}'$ such that, at least $P$ idle iterations took place during the interval $\left[t_{1},{t_{1}}' - 2C\right]$. 
The second phase starts at step $t_{2} = {t_{1}}' + 1$, and so on.

Consider two consecutive phases starting at steps $i$ and $j$ respectively.
We now prove that $P\left\{\Phi_{j} \leq \frac{7}{10}\Phi_{i}\right\} > \frac{1}{4}$.
Recall that we can partition the potential as $\Phi_{i} = \Phi_{i}\left(A_{i}\right) + \Phi_{i}\left(D_{i}\right)$.
Since, from the beginning of each phase and until its last $2C$ steps, at least $P$ idle iterations take place, then, by Lemma~\ref{lemma:phase potential decrease} (corresponding to Lemma~\ref{lemma:app-phase potential decrease} of the appendix) it follows
$P\left\{\Phi_{i} - \Phi_{j} \geq \frac{3}{10}\Phi_{i}\left(D_{i}\right)\right\} > \frac{1}{4}$.
Now, we have to prove the potential also drops by a constant fraction of $\Phi_{i} \left(A_{i}\right)$.
Consider some processor $p \in A_{i}$:
\begin{itemize}
  \item If $p$ does not have an assigned node, then $\Phi_{i}\left(p\right) = 0$.
  \item Otherwise, if $p$ has an assigned node $u$ at step $i$, then, $\Phi_{i}\left(p\right) = \phi_{i}\left(u\right)$.
  Noting that each phase has more than $C$ steps, then, $p$ executes $u$ before the next phase begins (\textit{i.e.}~before step $j$).
  Thus, the potential drops by at least $\frac{47}{64}\phi_{i}\left(u\right)$ during that phase.
\end{itemize}
Cumulatively, for each $p \in A_{i}$, it follows $\Phi_{i} - \Phi_{j} \geq \frac{47}{64}\Phi_{i}\left(A_{i}\right)$.
Thus, no matter how $\Phi_{i}$ is partitioned between $\Phi_{i}\left(A_{i}\right)$ and $\Phi_{i}\left(D_{i}\right)$, we have
$P\left\{\Phi_{i} - \Phi_{j} \geq \frac{3}{10}\Phi_{i}\right\} > \frac{1}{4}$.

We say a phase is successful if it leads the potential to decrease by at least a $\frac{3}{10}$ fraction.
So, a phase succeeds with probability at least $\frac{1}{4}$.
Since the potential is an integer, and, as aforementioned, starts at $\Phi_{0} = 4^{3T_{\infty} - 2}$ and ends at $0$, then, there can be at most
$\left(3T_{\infty} - 2\right)\log_{\frac{10}{7}}\left(4\right) < 12T_{\infty}$ successful phases.
If we think of each phase as a coin toss, where the probability that we get heads is at least $\frac{1}{4}$, then, the expected number of coins we have to toss to get heads $12T_{\infty}$ times is at most $48T_{\infty}$.
In the same way, the expected number of phases needed to obtain $12T_{\infty}$ successful ones is at most $48T_{\infty}$.
Consequently, the expected number of phases is $O\left(T_{\infty}\right)$.
Moreover, as each phase contains $O\left(P\right)$ idle iterations, the expected number of idle iterations is $O\left(PT_{\infty}\right)$.

Now, suppose the execution takes $n = 48T_{\infty} + m$ phases.
Each phase succeeds with probability greater or equal to $p = \frac{1}{4}$, meaning the expected number of successes is at least $np = 12T_{\infty} + \frac{m}{4}$.
We now compute the probability that the number of $X$ successes is less than $12T_{\infty}$.
We use the \emph{Chernoff bound}~\cite{DBLP:books/wi/AlonS92},
$P\left\{X < np - a\right\} < \euler^{-\frac{a^{2}}{2np}}$ with $a = \frac{m}{4}$.
It follows, $np - a = 12T_{\infty}$.
Choosing $m = 48T_{\infty} + 16\ln\left(\frac{1}{\varepsilon}\right)$, we have
$
P\left\{X < 12T_{\infty}\right\} < \euler^{-\frac{\left(\frac{m}{4}\right)^{2}}{2\left(12T_{\infty} + \frac{m}{4}\right)}}
\leq \euler^{-\frac{m}{16}}
\leq \euler^{-\frac{16\ln\left(\frac{1}{\varepsilon}\right)}{16}}
= \varepsilon
$.
Thus, the probability that the execution takes $96T_{\infty} + 16\ln\left(\frac{1}{\varepsilon}\right)$ phases or more, is less than $\varepsilon$.
With this we conclude that the number of idle iterations is at most
$O\left(\left(T_{\infty} + \ln\left(\frac{1}{\varepsilon}\right)\right)P\right)$
with probability at least $1 - \varepsilon$.
\end{proof}